\documentclass[a4paper,11pt]{article}
\usepackage{amsmath,a4wide}
\usepackage{slashed}
\usepackage{amssymb}
\usepackage{amscd}
\usepackage{bm}
\usepackage{epsfig}
\usepackage{graphics}
\usepackage{graphicx}
\usepackage{float}
\usepackage{epsfig}
\usepackage{amsmath}
\usepackage{graphics}
\usepackage{graphicx}
\usepackage{color}
\usepackage{epstopdf}
\usepackage{color}
\usepackage{booktabs}
\usepackage{subfigure}
\usepackage{ifpdf}
\usepackage{amssymb}
\usepackage{cite}
\usepackage{amssymb}
\usepackage{amsthm}
\usepackage{amsmath}
\usepackage[justification=centering]{caption}
\usepackage[]{natbib}

\usepackage{hyperref}
\hypersetup{
	colorlinks=true,
	linkcolor=blue,
	urlcolor=blue,
	citecolor=blue}

\usepackage[titletoc]{appendix}
\usepackage{verbatim}

\newtheorem{definition}{Definition}
\newtheorem{theorem}{Theorem}
\newtheorem{lemma}{Lemma}
\newtheorem{pro}{Proposition}

\newtheorem{assumption}{Assumption}
\begin{document}
	\title{Time-consistent mean-variance reinsurance-investment problem with long--range dependent mortality rate}
	
	\author{Ling Wang\thanks{Department of Statistics, The Chinese University of Hong Kong, Shatin, N.T., Hong Kong. \newline({\tt lingwang@link.cuhk.edu.hk})}
		\and Mei Choi Chiu\thanks{Department of Mathematics \& Information Technology, The Education University of Hong Kong, Tai Po, N.T., Hong Kong. \newline({\tt mcchiu@eduhk.hk})}	
		\and Hoi Ying Wong\thanks{Corresponding author. Department of Statistics, The Chinese University of Hong Kong, Shatin, N.T., Hong Kong. \newline({\tt hywong@cuhk.edu.hk})}
	}
	\date{\today}
	
	\maketitle \pagestyle{plain} \pagenumbering{arabic}
	\begin{abstract}
		This paper investigates the time-consistent mean-variance reinsurance-investment (RI) problem faced by life insurers. Inspired by recent findings that mortality rates exhibit long-range dependence (LRD), we examine the effect of LRD on RI strategies. We adopt the Volterra mortality model proposed in \cite{WCW} to incorporate LRD into the mortality rate process and describe insurance claims using a compound Poisson process with the intensity represented by stochastic mortality rate. Under the open-loop equilibrium mean-variance criterion, we derive explicit equilibrium RI controls and study the uniqueness of these controls in cases of constant and state-dependent risk aversion. We simultaneously resolve difficulties arising from unbounded non-Markovian parameters and sudden increases in the insurer's wealth process. We also use a numerical study to reveal the influence of LRD on equilibrium strategies. 
	\end{abstract}
	\vspace{2mm} {\em Keywords}: Mean--variance;  Time consistency; Reinsurance-investment;  Mortality model;  Long--range dependence  \\

	\section{Introduction}
	Insurers can manage their risk exposure through reinsurance and enhance their profits by investing in financial markets. These market-related practices inspire studies on optimal reinsurance and investment (RI) strategies. Many RI strategies are developed using the mean-variance (MV) criterion, which is popular in the field of practical investment. For instance, \cite{CY} consider an optimal RI problem involving an insurer and the MV criterion under a regime-switching model. \cite{SZ} study an optimal MV-RI problem with a delay using the maximum principle approach. \cite{YVLZ} investigate robust optimal RI strategies under a benchmarking MV criterion. 
	
	As acknowledged in the literature, the optimal MV problem is complicated by time inconsistency, which prevents satisfaction of the dynamic programming principle. Specifically, a strategy that is optimal at the initial time point becomes suboptimal at a later time point. Time inconsistency is a universal property of problems with MV objectives; accordingly, this potential weakness also appears in the corresponding RI problems. The first documentation of time inconsistency appears in \cite{S}. \cite{BC2010} propose to resolve problems complicated by time inconsistency by introducing an equilibrium feedback control framework based on the concept of sequential games. \cite{BMZ2014} argue that the state-independent strategy provided in \cite{BC2010} is economically unrealistic because the amount invested in the risky asset is independent of the investor's current wealth. They propose the concept of state-dependent risk aversion and establish an extended Hamilton--Jacobi--Bellman (HJB) framework for time inconsistency. However, \cite{BKM2017} show that it is technically difficult to study the uniqueness of equilibrium control within the HJB framework. 
	
	To simultaneously resolve time inconsistency and study uniqueness in MV problems, \cite{HJZ2012} introduce an open-loop equilibrium control framework involving a system of forward--backward stochastic differential equations (FBSDEs). \cite{HJZ2017} further provide sufficient and necessary conditions for equilibrium control and rigorously prove uniqueness under the assumption of bounded parameters in the wealth process. Extensions of this approach to cases involving jumps, stochastic volatility, and constraints on frameworks with controls and robust controls can be found in \cite{SG}, \cite{YW19}, \cite{HHLb} and \cite{HPW}, respectively.
	
	The concept of time-consistent (TC) MV formulation extends to studies of RI strategies. Related studies based on diffusion approximations of the claim process include, but are not limited to,  \cite{LRZ2015}, \cite{YVLZ}, \cite{HLV}, \cite{WS20}, \cite{YW}, and the references therein. Alternatively, studies that seriously consider jumps include \cite{LQ2016}, \cite{ZLG}, \cite{ACS}, \cite{SZY}, \cite{GW2020}, and the references therein. We follow the latter set of references and use jumps to describe the surplus process of the insurer and allow the intensity of the Poisson process, or the equivalent mortality rate, to be a stochastic process. The notable feature distinguishing this paper from the literature is that the mortality rate follows a stochastic process with long-range dependence (LRD), which is also known as the long-memory property or persistence.
	
	In a recent study using data from 16 countries, \cite{YPC} show empirically that mortality rate data exhibit LRD. Further, \cite{YPC2020} empirically confirm LRD in a multivariate time series of multi-cohort mortality data. \cite{DO} find strong empirical evidence for the existence of LRD in mortality data from an Italian population. Based on the empirical evidence, we investigate the impact of LRD on RI strategies, specifically the TCMV RI strategy. To the best of our knowledge, this paper is first to consider the TC RI problem using a mortality rate with LRD.
	
	Our investigation is based on the innovative Volterra mortality models (VMMs) proposed by \cite{WCW} for a mortality rate with LRD. Whereas the VMMs are tractable for actuarial valuation and longevity hedging, the non-Markovian and non-semimartingale features of the Volterra process generate subtle difficulties for us in deriving the open-loop equilibrium RI strategy. We also encounter jumps in the insurer's wealth process. Therefore, the main contribution of this paper lies in the its ability to overcome the aforementioned difficulties and derive unique explicit equilibrium controls for an RI problem under a VMM with both constant and state-dependent risk aversion TCMV objectives. Our work reinforces that the admissibility and uniqueness of equilibrium controls are non-trivial in a case involving an unbounded Volterra parameter that is simultaneously associated with the VMM and jumps. We further provide the first set of technical conditions and proofs of the admissibility and uniqueness of the equilibrium RI policy. 
	
	In our open-loop equilibrium framework, the key mathematical challenge is the simultaneous encounter of both the unbounded stochastic Volterra mortality rate and jumps in the state process. Although \cite{SG} extend their open-loop equilibrium framework to jump diffusion, their assumption regarding uniformly bounded parameters is too restrictive for our problem. \cite{ACS} study the open-loop RI problem under a jump diffusion model but require the Poisson intensity to be a bounded deterministic function. \cite{YW} study the open-loop RI problem with unbounded Markovian stochastic volatility, but their models do not include a jump term. We provide a necessary and sufficient condition for equilibrium control with jumps and unbounded Volterra parameters in the state process. We also carefully study the box constraint on the proportional reinsurance policy when the policy process falls in the interval of [0,1].
	
	From an actuarial science perspective, we provide the first rigorous proof that the equilibrium reinsurance policy is independent of the historical mortality rate under the TCMV criterion with constant risk aversion even though the mortality rate exhibits LRD. Such a reinsurance strategy is a unique equilibrium policy. This is a rather strong result because it confirms that LRD does not affect the (unique equilibrium) reinsurance demand of TCMV investors with constant risk aversion. In contrast, \cite{WW21} show that LRD has material effects on longevity pricing and hedging, implying that reinsurance is more robust to LRD in the mortality rate among TCMV investors with constant risk aversion. However, LRD does affect TCMV investors with state-dependent risk aversion. We therefore numerically examine this effect of LRD in the latter case.
	
	The remainder of this paper is organized as follows. In Section \ref{section:problem}, we describe our model of mortality with LRD and formulate an RI problem using the open-loop equilibrium control framework. In Section \ref{subsec:CRA}, we derive explicit forms for equilibrium strategies under constant risk aversion for both cases with the positive constraint and the $[0,1]$ interval constraint on the reinsurance policy, respectively. In Section \ref{Sec:SDRA}, we study the RI problem under the state-dependent risk aversion situation. The admissibility of the equilibrium controls together with some technical conditions. In Section \ref{section:numerical}, we use a numerical study to reveal the influence of the mortality rate with LRD on RI strategies. Our concluding remarks are presented in Section \ref{section:conclusion}.  Technical proofs are collected in the Appendix. 
	
	\section{Problem formulation}
	\label{section:problem}
	In a filtered complete probability space $(\Omega, \mathcal{F},\{\mathcal{F}\}_{t\in[0,T]},\mathbb{P})$, for any positive constant $q$ and positive integer $d$, let $\mathcal{D}$ be a nonempty subset of $\mathbb{R}^d$. The following terms are defined:
	\begin{itemize}
		\item[]$S_\mathcal{F}^q(t,T;\mathcal{D}, \mathbb{P})$: the set of all $\{\mathcal{F}_s\}_{s\in[t,T]}$-adapted $\mathcal{D}$-valued stochastic processes $X(\cdot)=\{X(s): t\leq s\leq T\}$ with $\mathbb{E}[{\rm sup}_{t\leq s\leq T}|X(s)|^q]<+\infty$;
		\item[]$L_\mathcal{F}^\infty(t,T;\mathcal{D}, \mathbb{P})$: the set of all essentially bounded
		$\{\mathcal{F}_s\}_{s\in[t,T]}$-adapted $\mathcal{D}$-valued stochastic processes;
		\item[] $L_\mathcal{F}^q(t,T;\mathcal{D}, \mathbb{P})$: the set of all
		$\{\mathcal{F}_s\}_{s\in[t,T]}$-adapted $\mathcal{D}$-valued stochastic processes with $\mathbb{E}\left[\int_{t}^{T}|X(s)|^qds\right]<+\infty$;
		\item[] $H_\mathcal{F}^q(t,T;\mathcal{D}, \mathbb{P})$: the set of all $\{\mathcal{F}_s\}_{s\in[t,T]}$-adapted $\mathcal{D}$-valued stochastic processes $X(\cdot)=\{X(s): t\leq s\leq T\}$ with $\mathbb{E}\left[\left(\int_{t}^{T}|X(s)|^2ds\right)^{q/2}\right]<+\infty$.
	\end{itemize}
	
	\subsection{Mortality model}
	Use $\hat{\lambda}(t)$ to denote the mortality rate of a population. Insurance claims arrive according to a Poisson process with an intensity proportional to $\hat{\lambda}(t)$. \cite{YPC} show that the mortality rates of different cohorts exhibit LRD, with an estimated Hurst parameter $H$ of around 0.8; in this estimation, the Hurst parameter is restricted to the interval (0,1). The mortality rate follows a Markovian process when $H=0.5$. Therefore, the empirical result reported by \cite{YPC} suggests significant LRD. Let $\mathbb{R}_+ = [0, \infty)$. To incorporate the LRD of the mortality rate, we adopt the VMM proposed in \cite{WCW}, so that 
	\begin{equation}\label{mortality}
		\hat{\lambda}_t = l(t) + \lambda_t, 
	\end{equation}
	where $l(t) \geq 0$ is a bounded deterministic function and $\lambda$ follows a Volterra process. Specifically, $\lambda_t$ follows a stochastic Volterra integral equation (SVIE):
	\begin{equation}\label{lambda}
		\lambda_t = \lambda_0+\int_{0}^{t}K(t-s)( b_1 - a_1\lambda_s)ds+\int_{0}^{t}K(t-s)\sigma_\lambda\sqrt{\lambda_s} dW_0(s),
	\end{equation}
	where $\lambda_0, b_1, a_1, \sigma_\lambda$ are positive constants, $W_0$ represents standard Brownian motion, and $K \in L^2(\mathbb{R}_+, \mathbb{R})$ is the Volterra kernel. When the kernel is set to the fractional kernel displayed in Table \ref{kernel}, then $\int_0^tK(t-s)dW_0(s)$ is proportional to classic fractional Brownian motion (fBM), a building block used in continuous-time models with LRD. The Volterra process is generally non-Markovian and non-semimartingale. Fortunately, the mortality model specified in \eqref{mortality} possesses an affine structure, as shown in \cite{WCW}, and thus is analytically tractable in terms of the Fourier--Laplace function. In this paper, we make the following standard assumption regarding the kernel $K(\cdot) \in L^2_{loc}(\mathbb{R}_+, \mathbb{R})$. 
	
	\begin{assumption}\label{assumption1}
		The kernel $K$ in \eqref{lambda} is strictly positive and completely monotone on $(0, \infty)$, and there exist $\chi \in (0,2]$ and $k > 0$ such that 
		\[\int_{0}^{h}K^2(t)dt + \int_{0}^{T}(K(t+h) - K(t))^2dt \leq kh^\chi, ~ h>0.\]
	\end{assumption}
	\begin{lemma}
		Under Assumption \ref{assumption1}, the SVIE \eqref{lambda} admits a unique in law $\mathbb{R}_+$-valued continuous weak solution for any initial condition $\lambda_0 \in \mathbb{R}_+$. (Theorem 6.1 in \cite{AJ}.)
	\end{lemma}
	Table \ref{kernel} offers some examples of kernel functions $K$ that satisfy Assumption \ref{assumption1}. Therefore, the fractional kernel satisfies our standing assumption. In fBM, the degree of LRD is reflected by the parameter $\alpha$, which is related to the classic Hurst parameter $H$ such that $\alpha = H-1/2$. Table \ref{kernel} also shows the resolvents corresponding to each example kernel. 
	
	\begin{table}[H]
		\centering
		\begin{tabular}{ccccc}  
			\toprule
			\toprule
			&Constant&Fractional&Exponential&Gamma \\
			\midrule
			\midrule
			$K(t)$& $c$ & $c\frac{t^{\alpha-1}}{\Gamma(\alpha)}$ & $ce^{-\lambda t}$  & $ce^{-\lambda t}\frac{t^{\alpha-1}}{\Gamma(\alpha)}$\\
			\midrule
			$r_t$& $ce^{-ct}$ & $ct^{\alpha-1}E_{\alpha,\alpha}(-ct^{\alpha})$&$ce^{-\lambda t}e^{-ct}$ & $ce^{-\lambda t}t^{\alpha-1}E_{\alpha,\alpha}(-ct^{\alpha})$\\
			\bottomrule
		\end{tabular}
		\caption{Examples of kernel function $K$ and the corresponding resolvent $R$. Here, $E_{\alpha, \beta}(z)=\sum_{n=0}^{\infty}\frac{z^{n}}{\Gamma(\alpha n+\beta)}$ denotes the Mittag-Leffler function.}
		\label{kernel}
	\end{table}
	
	The \textit{resolvent} or \textit{resolvent of the second kind} corresponding to $K$ shown in the table is defined as the kernel $R$, such that $K*R=R*K=K-R$. The convolutions $K*R$ and $R*K$, wherein $K$ is a measurable function on $\mathbb{R}_+$ and $R$ is a measure on $\mathbb{R}_+$ of locally bounded variation, are defined by 
	\[(K*R)(t) = \int_{[0,t]}K(t-s)R(ds),~~ (R*K)(t) = \int_{[0,t]}R(ds)K(t-s)\]
	for $t>0$.
	
	\begin{lemma}\label{lemma:expmu}
		\citep{AJ} If $\lambda$ follows the SVIE \eqref{lambda}, then for any $0\leq t \leq T$ and a constant $c_0$, there exists a solution $\psi\in L^2([0, T],\mathbb{R})$ to the Volterra--Riccati equation $\psi = (c_0 - a_1\psi  + \frac{1}{2}\sigma_{\lambda}^2\psi^2)*K$ such that 
		\begin{equation}\label{expmu}
			\mathbb{E}\left[\left.e^{c_0\int_{0}^{T}\lambda_sds}\right|\mathcal{F}_t\right] = \exp(Y_t(T)), 
		\end{equation}
		where 
		\begin{align}\label{Y}
			\begin{split}
				Y_t(T) &= Y_0 + \int_{0}^{t}\psi(T-s)\sigma_{\lambda}dW_1(s) - \frac{1}{2}\int_{0}^{t}\psi(T-s)^2\sigma_{\lambda}^2ds,\\
				Y_0 &= \int_{0}^{T}\left[c_0\lambda_0 + \psi(s)(b_1- a_1\lambda_0) + \frac{1}{2}\psi(s)^2\sigma_{\lambda}^2\right]ds.
			\end{split}
		\end{align}
		An alternative expression of $Y$ is 
		\begin{equation}\label{Y2}
			Y_t(T) = c_0\int_{0}^{T}\mathbb{E}\left[\lambda_s|\mathcal{F}_t\right]ds + \frac{1}{2}\int_{t}^{T}\psi(T-s)^2\sigma_{\lambda}^2ds, 
		\end{equation}
		where 
		\begin{equation}\label{expectation:mu}
			\mathbb{E}[\lambda_T|\mathcal{F}_t] = \left(1 - \int_{0}^{T}R_B(s)ds\right)\lambda_0 + b_1\int_{0}^{T}E_B(T-s)ds + \int_{0}^{t}E_B(T- s)\sigma_{\lambda}dW_0(s).
		\end{equation}
		Here, $B=-a_1$, $R_B$ is the resolvent of $-KB$, and $E_B = K- R_B*K$. 
	\end{lemma}

	\subsection{State process}
	Consider the classic risk process for an insurer's surplus. When the insurer makes no reinsurance or investment, his surplus process can be described as follows:
	\begin{equation}
		I(t) = I_0 + \int_{0}^{t}(1 + \theta)k_1\mu_z\hat{\lambda}_sds - \sum_{i =1}^{N(t)}z_i,
	\end{equation}
	where $I_0$ is the initial surplus and $\{z_i\}_{i=1}^{\infty}$ are  independent, identically distributed (iid), positive random variables representing insurance claims. $N(t)$ is a stochastic Poisson process with an intensity $k_1\hat{\lambda}_t$, where $k_1> 0$ is a constant. We assume that the insurance claims are independent of the mortality rate. $\tau_t = (1 + \theta)k_1\mu_z\hat{\lambda}_t$ is the premium rate, with $\mathbb{E}[z] = \mu_z$ and $\theta>0$ representing the safety loading of the insurer.  
	
	Suppose that the insurer is allowed to purchase reinsurance or acquire a new business. For any value of $t\in [0,T]$, denote the proportional reinsurance strategy by $a(t) \in [0, + \infty)$. $a(t) \in [0, 1]$ corresponds to a proportional reinsurance coverage. Therefore, when an insurance claim occurs, the reinsurance company pays a $1-a(t)$ fraction, while the insurer pays the remaining $a(t)$ fraction. Meanwhile, the reinsurance company charges the insurer at the rate of $\frac{1+ \eta}{1 + \theta}\tau_t(1-a(t))$, where $\eta\geq \theta$ represents the safety loading of the reinsurer. When $a(t)>1$, the insurer acquires for a new business. The surplus process of the insurer becomes 
	\begin{equation}
		dI(t) = [(\theta - \eta)k_1\hat{\lambda}_t\mu_z + (1+ \eta)k_1\hat{\lambda}\mu_za(t)]dt - a(t)d\sum_{i =1}^{N(t)}z_i.
	\end{equation}
	
	In practice, insurers also attempt to profit from the financial market. Consider a market consisting of risk-free and risky assets. The price of the risk-free asset, $B(t)$, is as follows:
	\[dB(t) = r_tB(t)dt,\]
	where the interest rate $r > 0$ is a bounded deterministic function. The price of the risky asset is as follows:
	\[dS(t) = S(t)[\mu(t) dt + \sigma(t) dW_1(t)],\]
	where $\mu(t), \sigma(t)> 0$ are two bounded deterministic functions, and $W_1$ represents standard Brownian motion independent of $W_0$.
	
	Suppose that the insurer has an initial wealth of $X_0$. For $t \in [0, T]$, denote $\pi(t)$ as the amount of money invested in the risky asset and $X(t)$ as the wealth process. The remaining amount of money, $X(t) - \pi(t)$, is then invested in the risk-free asset. Hence, we obtain
	\begin{equation*}
		dX_t = [r_tX_t + (\mu - r)\pi(t)  + (1 + \eta)k_1\hat{\lambda}_t\mu_za(t) + (\theta - \eta)k_1\hat{\lambda}_t\mu_z]dt + \pi(t)\sigma(t) dW_1(t) - a(t)d\sum_{i =1}^{N(t)}z_i.
	\end{equation*}
	Next, we simplify this notation by introducing a Poisson random measure. Denote the compound Poisson process by $\sum_{i=1}^{N(t)}z_i = \int_{0}^{t}\int_{\mathbb{R}_+}zN(ds,dz)$, where $\mathbb{R}_+ = [0, \infty)$ and $N(dt, dz)$ is a Poisson random measure in the space $[0,T]\times \mathbb{R}_+$.  
	Assume that the positive random claim size $z$ has a probability density function $f(z)$ with finite expectation and moments. Here, $\mu_z = \int_{\mathbb{R}_+} zf(z)dz$ and $\mathbb{E}[z^2] =\int_{\mathbb{R}_+} z^2f(z)dz$.  Then the Poisson random measure $N(dt , dz)$ has the  compensator $\delta(dz)dt \triangleq k_1\hat{\lambda}_tf(z)dzdt$, where $\delta(dz) =  k_1\hat{\lambda}_tf(z)dz$. Define $F^2(t, T; \mathbb{R})$ as the set of all $\{\mathcal{F}_s\}_{s \in[t, T]}$-predictable processes $X(\cdot, \cdot): \Omega\times[t, T]\times\mathbb{R}_+\rightarrow \mathbb{R}$, such that $\mathbb{E}\left[\int_{t}^{T} \|X(s, \cdot)\|^2_{L^2}ds\right] < \infty$. Here, $\|X(s, z)\|^2_{L^2} := \int_{\mathbb{R}_+}|X(s, z)|^2\delta(dz)$. 
	Let $\widetilde{N}(dt,dz) = N(dt,dz)- \delta(dz)dt$. The insurer's wealth process becomes
	\begin{align}
		dX_t &= [r_tX_t + (\mu - r)\pi(t)+  \eta k_1\hat{\lambda}_t\mu_za(t)  + (\theta - \eta)k_1\hat{\lambda}_t\mu_z ]dt + \pi(t)\sigma(t)dW_1(t)\notag \\
		&-\int_{\mathbb{R}_+} a(t)z\widetilde{N}(dt,dz).\notag
	\end{align}
	For convenience, we rewrite the wealth process as follows:
	\begin{equation}\label{X}
		dX_t = [r_tX_t + \nu(t)^\top u(t)  + c_t ]dt + \pi(t)\sigma(t) dW_1(t) -\int_{\mathbb{R}_+} a(t)z\widetilde{N}(dt,dz),
	\end{equation}
	where $u(t) = (\pi(t), a(t))^\top$ is the control pair; $\nu = (\nu_1, \nu_2)^\top$ with $\nu_1(t) = \mu(t) - r_t$ and $\nu_2(t) = \eta k_1\hat{\lambda}_t\mu_z$; and $c_t = (\theta - \eta)k_1\hat{\lambda}_t\mu_z$. The coefficients $\nu_2(t)$ and $c_t$ are unbounded stochastic processes due to the unbounded parameter $\hat{\lambda}_t$. Note the existence of a constraint specifies that the control $a(t)$ should always be nonnegative.  
	
	\subsection{Open-loop equilibrium framework}
	\begin{definition}\label{admissible}
		If the control $u(\cdot) \in H_\mathcal{F}^2(0,T;\mathbb{R}, \mathbb{P}) \times  \cup_{q>2}L_\mathcal{F}^q(0,T;\mathcal{D}, \mathbb{P})$ and Equation \eqref{X} admits a unique strong solution $X \in S_\mathcal{F}^2(0,T;\mathbb{R}, \mathbb{P})$, then $u$ is called an admissible control. 
	\end{definition}
	Two cases are considered in this paper.
	\begin{enumerate}
		\item[\textbullet]$\mathcal{D} = \mathbb{R}_+$. The reinsurance strategy is only required to be positive. 
		\item[\textbullet] $\mathcal{D} = [0, 1]$. The reinsurance strategy is restricted to $[0, 1]$. 
	\end{enumerate}
	
	\begin{lemma}\label{lemma1}
		Let $\lambda(\cdot)$ be the continuous solution to the SVIE \eqref{lambda}. Suppose that kernel $K$ satisfies Assumption 1; then, there exists a constant $C>0$ such that 
		\[ \sup_{0\leq t\leq T}\mathbb{E}[|\lambda_t|^q] < C. \]
		for any constant $q\geq 2$. 
	\end{lemma}
	\begin{proof}
		In \eqref{lambda}, it is clear that 
		\[ |b_1- a_1\lambda| \vee |\sigma_{\lambda}\sqrt{\lambda}| \leq C_1(1+ |\lambda|),\]
		for a constant $C_1$. The result follows from Lemma 3.1 in \cite{AJ}. 
	\end{proof}
	
	The insurer's objective is to minimize 
	\begin{equation}\label{objective}
		J(t,x_t; u(\cdot)) = \frac{1}{2}{\rm Var}_t(X_T) -  (\phi_1x_t + \phi_2)\mathbb{E}_t[X_T] = \frac{1}{2}\left(\mathbb{E}_t[X_{T}^2] - \mathbb{E}_t^2[X_{T}]\right) -  (\phi_1x_t + \phi_2)\mathbb{E}_t[X_{T}]
	\end{equation}
	by an admissible control $u$, where $x_t = X(t)$, $\mathbb{E}_t[\cdot] = \mathbb{E}[\cdot|\mathcal{F}_t]$, and $\phi_1, \phi_2 \in \mathbb{R}_+$.  If $\phi_1 = 0$ and $\phi_2> 0$, the insurer has a constant risk aversion; otherwise, if $\phi_1 > 0$, the insurer has a state-dependent risk aversion. 
	
	Because of the unboundedness of the parameter $\hat{\lambda}$, the results in \cite{SG} cannot be applied directly to our problem. In our case, the admissibility of the control is highly nontrivial, especially for the state-dependent case. The following theorem and the proof detailed in Appendix \ref{Appendix:remarkproof} are useful in this regard. 
	\begin{theorem}\label{remark:admissible}
		If the control $u(\cdot) \in H_\mathcal{F}^2(0,T;\mathbb{R}, \mathbb{P})\times \cup_{q>2}L_\mathcal{F}^q(0,T;\mathcal{D}, \mathbb{P}) $,  then  $X \in S_\mathcal{F}^2(0,T;\mathbb{R}, \mathbb{P})$.  
	\end{theorem}
	Theorem \ref{remark:admissible} asserts that the condition for the control space in Definition \ref{admissible} alone is sufficient to ensure admissibility because the condition for the state process is immediate.
	
	We employ the open-loop equilibrium framework in \cite{HJZ2012, HJZ2017} and \cite{HHLb}. Given  a control pair $(\pi^*(t), a^*(t)) \in H_\mathcal{F}^2(0,T;\mathbb{R}, \mathbb{P})\times \cup_{q>2}L_\mathcal{F}^q(0,T;\mathcal{D}, \mathbb{P})$, for any $t \in [0, T)$, $\epsilon > 0$, $\rho_1 \in H_\mathcal{F}^2(t,T;\mathbb{R}, \mathbb{P})$, and $\rho_2 \in \cup_{q>2}L_\mathcal{F}^q(t,T;\mathcal{D}, \mathbb{P})$, define
	\begin{align}\label{perturbation}
		\begin{split}
			\pi_s^{t, \epsilon, \rho_1} &= \pi^*_s + \rho_1 \boldsymbol{1}_{s \in [t, t+\epsilon)}, ~ s\in [t, T],\\
			a_s^{t, \epsilon, \rho_2} &= a_s^* + ( \rho_2(s) -  a_s^*)\boldsymbol{1}_{s \in [t, t+\epsilon)}, ~ s\in [t, T]. 
		\end{split}
	\end{align}
	Under this construction, for any $t\in [0, T)$, we obtain $a_s^{t, \epsilon, \rho_2} =  \rho_2(s) \geq 0$, when $s \in [t, t+\epsilon]$, and $a_s^{t, \epsilon, \rho_2} = a^*_s \geq 0$, when $s \in [t + \epsilon, T]$. A similar setting is used in \cite{HHLb} and \cite{YW}.
	
	\begin{definition}\label{def2}
		Let $(\pi^*(t), a^*(t)) \in H_\mathcal{F}^2(0,T;\mathbb{R}, \mathbb{P})\times \cup_{q>2}L_\mathcal{F}^q(0,T;\mathcal{D}, \mathbb{P})$ be a given control pair  and $X^*$ be the corresponding state process. Then, the control pair $(\pi^*(t), a^*(t))$ is an equilibrium strategy for problem \eqref{objective} if, for any $t\in [0, T)$, $\rho_1 \in H_\mathcal{F}^2(0,T;\mathbb{R}, \mathbb{P})$, and $\rho_2 \in \cup_{q>2}L_\mathcal{F}^q(t,T;\mathcal{D}, \mathbb{P})$, we have  $X^* \in S_\mathcal{F}^2(t,T;\mathbb{R}, \mathbb{P})$ and
		\begin{equation}\label{liminf}
			\liminf_{\epsilon \downarrow 0} \frac{J(t, X^*_t; \pi^{t, \epsilon, \rho_1}, a^{t, \epsilon, \rho_2}) - J(t, X^*_t; \pi^*, a^*)}{\epsilon} \geq 0,
		\end{equation}
		where $\pi^{t, \epsilon, \rho_1}, a^{t, \epsilon, \rho_2}$ is as defined in \eqref{perturbation}. 
	\end{definition}
	
	Note that we consider the constraint on the reinsurance control to be in line with \cite{HHLb}. When the control $a^*$ is restricted to $\mathcal{D}$, then $\rho_2$ is also restricted to $\mathcal{D}$ in the definition. 
	
	The nature of the Volterra process prevents the use of a classic HJB framework in our problem. Inspired by \cite{HJZ2012, HJZ2017} and \cite{SG}, we adopt the BSDE approach and provide an equivalent condition to \eqref{liminf} as follows. 
	
	For any $t\in[0, T]$, define the adjoint process $(p^*(s; t), Z^*(s; t), Z^*_2(s,z;t)) \in L_\mathcal{F}^2(t,T;\mathbb{R}, \mathbb{P})\times H_\mathcal{F}^2(t,T;\mathbb{R}^2, \mathbb{P})\times F^2(t,T;\mathbb{R})$ that satisfies the following BSDE:
	\begin{align}\label{p*}
		\left\{
		\begin{array}{lr}
			dp^*(s; t) = -r_sp^*(s; t)ds + Z^*(s;t)^\top dW_s + \int_{\mathbb{R}_+} Z_2^*(s,z;t)\widetilde{N}(ds,dz),\\
			p^*(T; t) = X^*_T - \mathbb{E}_t[X^*_T] -(\phi_1X^*_t + \phi_2), \end{array}
		\right.
	\end{align}
	where $Z^*(s;t) = (Z_0^*(s;t), Z_1^*(s;t))^\top$, $W = (W_0, W_1)^\top$,  and $X^*$ is the state process corresponding to the control $(\pi^*(t), a^*(t)) \in H_\mathcal{F}^2(0,T;\mathbb{R}, \mathbb{P})\times \cup_{q>2}L_\mathcal{F}^q(0,T;\mathcal{D}, \mathbb{P})$. The flow of BSDEs in \eqref{p*} is constructed to perturb $J$ in \eqref{liminf} and thus obtain the leading order term of the $\liminf$ for a small value of $\epsilon$ when the terminal condition is set to match the objective function. Although \cite{SG} extend the perturbation in \cite{HJZ2012, HJZ2017} to incorporate random jumps, in our case, the unboundedness of the Volterra process generates extra difficulties. Still, we prove the following theorem in Appendix \ref{appendix:Theo2prop1}. 
	
	\begin{theorem}\label{theorem1}
		For any $t \in [0,T]$, $\epsilon>0$, $\rho_1 \in H_\mathcal{F}^2(t,T;\mathbb{R}, \mathbb{P})$, and $\rho_2 \in \cup_{q>2}L_\mathcal{F}^q(t,T;\mathcal{D}, \mathbb{P})$, let  $(\pi^*(t), a^*(t)) \in H_\mathcal{F}^2(0,T;\mathbb{R}, \mathbb{P})\times \cup_{q>2}L_\mathcal{F}^q(0,T;\mathcal{D}, \mathbb{P})$ be a given control pair and  $\pi^{t, \epsilon, \rho_1}, a^{t, \epsilon, \rho_2}$ be as defined in \eqref{perturbation}. Then,    
		\[J(t, X^*_t; \pi_s^{t, \epsilon, \rho_1}, a_s^{t, \epsilon, \rho_2}) - J(t, X^*_t; \pi^*, a^*)= \mathbb{E}_t\int_{t}^{t + \epsilon} \left[\langle \Lambda(s; t), \rho_s \rangle + \langle \Theta(s)\rho_s, \rho_s \rangle\right] ds + o(\epsilon),\]
		where $\rho_s = (\rho_1, \rho_2- a^*_s)^\top$, 
		\begin{equation}\label{Lambda}
			\Lambda(s; t) = \left(\nu_1p^*(s;t) + \sigma(s)Z^*_1(s;t), \nu_2p^*(s;t) - \int_{\mathbb{R}_+} zZ^*_2(s;z,t)\delta(dz)\right)^\top, 
		\end{equation}
		and $\Theta(s) =\frac{1}{2}e^{\int_{s}^{T} 2r_u du }\left(\sigma(s)^2 + k_1\hat{\lambda}_s\mathbb{E}[z^2]\right)$ with $\left(p^*(s;t), Z^*_1(s;t), Z^*_2(s;z,t) \right)$ as defined in \eqref{p*}. 
	\end{theorem}
	
	Theorem \ref{theorem1} asserts that the leading order term of $\liminf$ consists of two parts. The first part is a functional of $\Lambda$, and the second is a functional of $\Theta$. Then, it is sufficient to deduce an equilibrium control $u^*$ that makes both parts nonnegative. From this expression, however, it is clear that $\Theta \geq 0$. Therefore, we present the following proposition. 
	\begin{pro}\label{prop1}
		Under the same assumption applied to Theorem \ref{theorem1},  $(\pi^*(t), a^*(t)) \in H_\mathcal{F}^2(0,T;\mathbb{R}, \mathbb{P})\times \cup_{q>2}L_\mathcal{F}^q(0,T;\mathcal{D}, \mathbb{P})$ is an equilibrium control if 
		\[ \liminf_{ \epsilon \downarrow 0}\frac{1}{\epsilon}\int_{t}^{t + \epsilon} \mathbb{E}_t[\langle \Lambda(s;t ), \rho_s \rangle]ds\geq 0, ~ a.s.,~ \forall t \in [0,T] \] 
		for any  $\rho_1 \in H_\mathcal{F}^2(t,T;\mathbb{R}, \mathbb{P})$ and $\rho_2 \in \cup_{q>2}L_\mathcal{F}^q(t,T;\mathcal{D}, \mathbb{P})$, where $\Lambda(s; t)$ is defined as in \eqref{Lambda}, and $\rho_s = (\rho_1, \rho_2- a^*_s)^\top$.
	\end{pro}
	
	In other words, it is sufficient for us to consider the non-negativity of the functional of $\Lambda$, which involves $(p^*, Z^*, Z_2^*)$ as a solution of BSDE \eqref{p*}. In fact, $\Lambda$ can be expressed more explicitly as follows; the proof is given in Appendix \ref{proofP2}.
	\begin{pro}\label{prop2}
		Let $(\pi^*(t), a^*(t)) \in H_\mathcal{F}^2(0,T;\mathbb{R}, \mathbb{P})\times \cup_{q>2}L_\mathcal{F}^q(0,T;\mathcal{D}, \mathbb{P})$ be a given control pair, and $X^*$ be the corresponding state process. For any $t_1, t_2\in[0, T]$, the unique solution to BSDE \eqref{p*} satisfies $Z^*(s;t_1) = Z^*(s;t_2)$ and $Z^*_2(s, \cdot;t_1) = Z^*_2(s,\cdot;t_2)$ for $a.e. ~ s \geq \max(t_1, t_2)$. Moreover, there exists a stochastic process $\Lambda_0$ valued in $\mathbb{R}^2$ and $\xi \in S_\mathcal{F}^2(t,T;\mathbb{R}, \mathbb{P})$ such that 
		\[ \Lambda(s; t) = \Lambda_0(s) + e^{\int_{s}^{T}r_udu}\nu(s)\xi(t).\]
	\end{pro}
	
	By substituting the expression of $\Lambda$ from Proposition \ref{prop2} into Proposition \ref{prop1}, we obtain a sufficient and necessary condition for equilibrium control in the following theorem. This condition is useful for deducing an explicit solution to the corresponding equilibrium control.
	
	\begin{theorem}\label{theorem2}
		Let $(p^*(s; t), Z^*(s; t), Z^*_2(s,z;t)) \in L_\mathcal{F}^2(t,T;\mathbb{R}, \mathbb{P})\times H_\mathcal{F}^2(t,T;\mathbb{R}^2, \mathbb{P})\times F^2(t,T;\mathbb{R})$ be the unique solution to BSDE \eqref{p*},  then  $(\pi^*(t), a^*(t)) \in H_\mathcal{F}^2(0,T;\mathbb{R}, \mathbb{P})\times \cup_{q>2}L_\mathcal{F}^q(0,T;\mathcal{D}, \mathbb{P})$  is an open-loop equilibrium control if and only if, for any   $\rho_2 \in \cup_{q>2}L_\mathcal{F}^q(t,T;\mathcal{D}, \mathbb{P})$, 
		\begin{align}\label{condition}
			\left\{
			\begin{array}{lr}
				\nu_1(t)p^*(t; t) + \sigma(t)Z^*_1(t;t) = 0,~ a.e.~ t \in [0, T], a.s.,\\ 
				\langle\nu_2(t)p^*(t; t) - \int_{\mathbb{R}_+}z Z^*_2(t, z; t)\delta(dz),  \rho_2(t) - a^*_t\rangle \geq 0~ a.e.~ t \in [0, T], a.s..
			\end{array}\right.
		\end{align}
	\end{theorem}
	Although the proof of Theorem \ref{theorem2} is presented in Appendix \ref{appendix:theorem3}, we explain here the use of this theorem to derive an admissible equilibrium control $(\pi^*, a^*)$. In this case, we must first solve for the flow of BSDEs in \eqref{p*} for $(p^*, Z^*, Z_2^*)$ upon the existence of a solution. Here, Theorem \ref{theorem2} is useful for verifying a potential solution to $p^*$. Once an appropriate solution is determined, Theorem \ref{theorem2} is also useful for extracting the corresponding equilibrium control $(\pi^*, a^*)$. Although this procedure is similar to others reported in the literature, we encounter the Volterra mortality rate when applying Theorem \ref{theorem2}. Interestingly, this rate has no effect on the case of constant risk aversion ($\phi_1 = 0$), such that the classic equilibrium control also works in the VMM. We further prove for the first time that such an equilibrium is unique. However, the proof under our consideration is non-trivial, although the equilibrium control agrees with the classic one without LRD. The state-dependent case ($\phi_1>0$) presents an additional challenge.
	
	\section{Equilibrium strategy under constant risk aversion}\label{subsec:CRA}
	Constant risk aversion refers to a setting wherein $\phi_1 = 0$ and $\phi_2 \geq 0$ in \eqref{p*}. Here, let $\mathcal{D} = \mathbb{R}_+$. Note that the reinsurance strategy should obey the non-negative constraint. When $\phi_1=0$, the terminal condition in \eqref{p*} does not involve $X^*_t$; therefore, the derivation procedure is very similar to the one discussed in \cite{YW}. Fortunately, in the case with a zero $\phi_1$, the derivation procedure eliminates the Volterra mortality rate from the solution form considered below in \eqref{pform}. As the BSDE in \eqref{p*} is linear, the solution is unique. 
	
	Inspired by \cite{YW}, consider the form of the solution to $p^*(s;t)$ as follows:
	\begin{equation}\label{pform}
		p^*(s;t) = M_sX^*_s + \Gamma_s^{(2)} - \mathbb{E}_t[M_sX^*_s + \Gamma_s^{(2)} - \Phi_s], 
	\end{equation}
	where $M_s = e^{\int_{s}^{T}2r_udu}$, $\Phi_s = -\phi_2e^{\int_{s}^{T}r_udu}$, and $(\Gamma_s^{(2)},\gamma_s^{(2)})$ is the solution to the following BSDE: 
	\begin{align}\label{BSDEgamma}
		d\Gamma_s^{(2)} = - F^{(2)}_sds + \gamma^{(2)}_sdW_s, ~ \Gamma^{(2)}_T = -\phi_2 .
	\end{align}
	Here, the driver $F_s^{(2)}$ and $\gamma_s^{(2)} = (\gamma^{(2)}_0(s), \gamma^{(2)}_1(s))^\top$ remain to be determined. For a fixed $t$, we apply the It\^o lemma, with jumps at time point $s$, to the solution form \eqref{pform}, and then compare the coefficients of the resulting SDE to those in \eqref{p*}. This enables us to identify the driver $F_s^{(2)}$.
	
	After substituting the result obtained for $p^*$ with an appropriate driver $F_s^{(2)}$ into Theorem \ref{theorem2}, we deduce that
	\begin{align}
		\pi_s^* &= - \frac{1}{M_s\sigma^2}(\nu_1\Phi_s + \sigma \gamma_1^{(2)}(s))= - \frac{(\mu(s) - r_s)\Phi_s}{M_s\sigma^2} - \frac{\gamma_1^{(2)}(s)}{M_s\sigma}, \label{pi*}\\
		a^*_s &= -\frac{\nu_2\Phi_s}{k_1\hat{\lambda}_sM_s\mathbb{E}[z^2]} =  -\frac{\eta \mu_z\Phi_s}{M_s\mathbb{E}[z^2]} > 0, \label{a*}
	\end{align}
	where $(\Gamma^{(2)}, \gamma^{(2)})$ is the solution to the BSDE as follows:
	\begin{align}\label{BSDEgamma2}
		\left\{
		\begin{array}{lr}
			d\Gamma^{(2)}_s = -\left\{r_s\Gamma^{(2)}_s - \frac{\nu_1}{\sigma} \gamma_1^{(2)}(s)- \left(\frac{\nu_1^2}{\sigma^2} + \frac{\nu_2^2}{k_1\hat{\lambda}_s\mathbb{E}[z^2]}\right)\Phi_s + M_sc_s\right\}ds+ \gamma_s^{(2)}dW_s, \\
			\Gamma_T^{(2)} = -\phi_2.
		\end{array}\right.
	\end{align} 
	As $\hat{\lambda}$ is the only stochastic unbounded coefficient in the BSDE \eqref{BSDEgamma2}, we consider the case in which $\gamma_1^{(2)} = 0$ and obtain a BSDE with a unique solution. 
	
	\begin{pro}\label{pro:gam0}
		At a set $\gamma_1^{(2)} = 0$, the BSDE \eqref{BSDEgamma2} admits a unique solution  $(\Gamma^{(2)}, \gamma^{(2)}_0)\in S_\mathcal{F}^q(0,T;\mathbb{R}, \mathbb{P})\times H_\mathcal{F}^q(0,T;\mathbb{R}, \mathbb{P})$ for any $q >2$. 
	\end{pro}
	\begin{proof}
		Recall that $\nu_1, r$, and $\gamma$ are deterministic functions;  $\frac{\nu_2^2}{k_1\hat{\lambda}\mathbb{E}[z^2]} = \frac{\eta^2\mu_z^2}{\mathbb{E}[z^2]}k_1\hat{\lambda}$; and $c = (\theta -\eta)\mu_zk_1\hat{\lambda}$.  According to H\"older's inequality and Lemma \ref{lemma1}, there exists a constant $C > 0$, such that  
		\[\mathbb{E}\left[\left(\int_{0}^{T}|\hat{\lambda}_s|ds\right)^q\right] \leq C \int_{0}^{T}\mathbb{E}\left[|\hat{\lambda}_s|^q\right]ds  < \infty,\]
		for any constant $q>2$. The result follows according to Theorem 5.1 in \cite{EPQ}.
	\end{proof}
	
	From \eqref{pi*} and \eqref{a*}, it is clear that $(\pi^*(t), a^*(t)) \in H_\mathcal{F}^2(0,T;\mathbb{R}, \mathbb{P})\times \cup_{q>2}L_\mathcal{F}^q(0,T;\mathbb{R}_+, \mathbb{P})$ when $\gamma_1^{(2)} = 0$. According to Theorem \ref{remark:admissible}, we obtain the following equilibrium reinsurance-investment strategy. 
	\begin{theorem}\label{ThmSolConstant}
		Let an explicit equilibrium control to Problem \eqref{objective} with a constant risk aversion $\phi_2>0$ be given as  
		\begin{equation}\label{u*c}
			\pi^*_s = \frac{\mu(s)- r_s}{\sigma(s)^2}\phi_2 e^{-\int_{s}^{T}r_udu}, ~ a^*_s = \frac{\eta\mu_z}{\mathbb{E}[z^2]}\phi_2e^{-\int_{s}^{T}r_udu}>0. 
		\end{equation}
	\end{theorem}
	Although we obtain an equilibrium reinsurance-investment strategy in Theorem \ref{ThmSolConstant}, it remains unclear whether this equilibrium is unique and whether the upper bound of the reinsurance strategy is 1. The latter boundedness condition eliminates the possibility that the insurer acquires a new (reinsurance) business. 
	
	\subsection{The $[0,1]$ interval constraint on reinsurance}
	\label{section:constrain}
	If we impose $a_t\in [0,1]$ as a bounded constraint on the reinsurance strategy, then we set $\mathcal{D} = [0, 1]$ in Theorems \ref{remark:admissible} to \ref{theorem2} and reiterate the derivation procedure over such a constraint, using \eqref{p*} and Theorem \ref{theorem2}. 
	
	Consider the solution form for $p^*(s;t)$, as in \eqref{pform}. For a fixed $t$,  we apply the It\^o formula at $s$ and find that $p^*(t; t) = \Phi(t)$, $Z_0^* = \gamma_0^{(2)}(t)$, $Z_1^*(t; t) = M(t)\pi^*(t)\sigma(t) + \gamma_1^{(2)}(t)$, and $Z_2^*(t, z; t) = -M(t)a^*(t)z$. Furthermore, we set $\gamma_1^{(2)}(t) = 0$. According to \eqref{condition},  
	\begin{align*}
		&\pi_s^* = - \frac{1}{M_s\sigma^2}\nu_1\Phi_s;\\
		&\langle \nu_2(t)\Phi(t) +M(t) \mathbb{E}[z^2]k_1\hat{\lambda}_ta^*(t), \rho_2(t)-a^*(t) \rangle \geq 0. 
	\end{align*}
	for any  $\rho_2 \in \cup_{q>2}L_\mathcal{F}^q(t,T;\mathcal{D}, \mathbb{P})$. Hence, we deduce that 
	\begin{equation}
		a^*(t) = {\rm Proj}_\mathcal{D}\left(-\frac{\nu_2\Phi_t}{k_1\hat{\lambda}_tM_t\mathbb{E}[z^2]}  \right) = {\rm Proj}_\mathcal{D}\left(-\frac{\eta \mu_z\Phi_t}{M_t\mathbb{E}[z^2]}\right), 
	\end{equation}
	where ${\rm Proj}_\mathcal{D}(x)$ represents the projection of a point $x \in \mathbb{R}$ onto $\mathcal{D}$. We refer to \cite{HL} for detailed information about projections.
	Thus, $(\Gamma^{(2)}, \gamma^{(2)}_0)$ becomes the solution to the BSDE.
	\begin{align}\label{BSDE:gamma22}
		\left\{
		\begin{array}{lr}
			d\Gamma_s^{(2)} = -\left\{r_s\Gamma_s^{(2)} - \frac{\nu_1^2}{\sigma^2}(s) +M_s\nu_2(s){\rm Proj}_\mathcal{D}\left(-\frac{\eta \mu_z\Phi_t}{M_t\mathbb{E}[z^2]}\right) +M_sc_s \right\}ds + \gamma^{(2)}_0(s)dW_0(s), \\
			\Gamma_T^{(2)} = -\phi_2. 
		\end{array}
		\right.
	\end{align}
	
	\begin{pro}
		The BSDE \eqref{BSDE:gamma22} admits a unique solution, $(\Gamma^{(2)}, \gamma^{(2)}_0)\in S_\mathcal{F}^q(0,T;\mathbb{R}, \mathbb{P})\times H_\mathcal{F}^q(0,T;\mathbb{R}, \mathbb{P})$, for any $q >2$. 
	\end{pro}                                                              
	
	\begin{proof}
		Recall that $\nu_2 = \eta\mu_zk_1\hat{\lambda}$, and $c = (\theta -\eta)\mu_zk_1\hat{\lambda}$. Lemma \ref{lemma1} clearly demonstrates that $\hat{\lambda} \in H_\mathcal{F}^q(0,T;\mathbb{R}, \mathbb{P})$ for any $q>2$. The result follows according to Theorem 5.1 in \cite{EPQ}.
	\end{proof}
	
	As both $\pi_s^*$ and $a^*$ are bounded, $(\pi^*(t), a^*(t)) \in H_\mathcal{F}^2(0,T;\mathbb{R}, \mathbb{P})\times \cup_{q>2}L_\mathcal{F}^q(0,T;\mathcal{D}, \mathbb{P})$. The following theorem is immediate.
	
	\begin{theorem}
		An admissible equilibrium control on Problem \eqref{objective}, with constraint $[0,1]$ on the reinsurance control and a constant risk aversion $\phi_2>0$, is given by 
		\begin{equation}\label{u*constrain}
			\pi^*_s = \frac{\mu(s)- r_s}{\sigma(s)^2}\phi_2 e^{-\int_{s}^{T}r_udu}, ~ a^*_s = {\rm Proj}_\mathcal{D}\left( \frac{\eta\mu_z}{\mathbb{E}[z^2]}\phi_2e^{-\int_{s}^{T}r_udu}\right). 
		\end{equation}
	\end{theorem}
	
	The strategies in \eqref{u*constrain} are implemented in a straightforward manner. In fact, the investment-reinsurance strategy $(\pi^*_t, a^*_t)$ is exactly the same as that in Theorem \ref{ThmSolConstant} when $a^*_t \le 1$. Otherwise, if $a^*_t \le 1$ in Theorem \ref{ThmSolConstant}, we obtain exactly the same $\pi^*$ but must return $a^*_t$ to 1 to fulfill the upper bound. The uniqueness of this equilibrium strategy remains to be demonstrated. 
	
	\subsection{Uniqueness of the equilibrium control}
	\begin{theorem}\label{uniqueCRA}
		Under Assumption \ref{assume:lambda}, the control $u^* = (\pi^*, a^*)$, given by \eqref{u*c}, is the unique equilibrium control for the RI problem \eqref{objective} when $\phi_1 = 0$.
	\end{theorem}
	\begin{proof}
		Suppose that an alternative admissible equilibrium control pair, $(\pi, a)$, exists with the corresponding state process $X$. By replacing $X^*$ with $X$, the BSDE \eqref{p*} admits a unique solution $(p(s; t), Z(s; t), Z_2(s,z;t))$. This satisfies the condition \eqref{condition}, wherein $Z(s; t) =\left(Z_0(s; t), Z_1(s; t)\right)^\top$.  Define
		\begin{align*}
			& \bar{p}(s;t) = p(s; t) - \left(M_sX_s + \Gamma_s^{(2)} - \mathbb{E}_t[M_sX_s + \Gamma_s^{(2)} - \Phi_s]\right), \\
			&\bar{Z}_0(s;t) = Z_0(s;t)- \gamma^{(2)}_0(s),~\bar{Z}_1(s; t) = Z_1(s;t) - M_s \pi(s)\sigma(s) ,\\
			&\bar{Z}_2(s,z;t) = Z_2(s,z;t) + M_sa(s)z,
		\end{align*}
		using the definitions of $M$, $\Gamma^{(2)}$, $\Phi$, and $\gamma^{(2)}_0$ provided in Section \ref{subsec:CRA}. 
		
		According to Proposition \ref{pro:gam0}, $(\bar{p}(s; t), \bar{Z}(s; t), \bar{Z}_2(s,z;t)) \in L_\mathcal{F}^2(t,T;\mathbb{R}, \mathbb{P})\times H_\mathcal{F}^2(t,T;\mathbb{R}^2, \mathbb{P}) \times F^2(t,T;\mathbb{R})$. Substituting this into \eqref{condition}, we obtain
		\begin{align}
			\left\{
			\begin{array}{lr}
				\nu_1(t)\left[\bar{p}(t; t) +\Phi_t\right] + \sigma(t)\left[\bar{Z}_1(t; t) + M_t \pi(t)\sigma(t)  \right] = 0,\\ 
				\langle\nu_2(t)\left[\bar{p}(t; t) +\Phi_t\right]  - k_1\hat{\lambda}_t\int_{\mathbb{R}_+}z_tf(z) \left[\bar{Z}_2(t, z; t)- M_ta(t)z\right]dz, \rho_2(t) - a_t\rangle \geq 0
			\end{array}\right.
		\end{align}
		for any $\rho_2 \in \cup_{q>2}L_\mathcal{F}^q(t,T;\mathbb{R}_+, \mathbb{P})$. Therefore, 
		\begin{align*}
			\pi(t) &= -\frac{\nu_1}{M_t\sigma(t)^2} \Phi_t - \frac{1}{M_t\sigma(t)^2}(\nu_1(t)\bar{p}(t; t) + \sigma(t)\bar{Z}_1(t; t))\\
			& = \pi^*(t) - \frac{1}{M_t\sigma(t)^2}(\nu_1(t)\bar{p}(t; t) + \sigma(t)\bar{Z}_1(t; t)) \triangleq  \pi^*(t) + D^c_1(t), \\
			a(t) &= \frac{1}{k_1\hat{\lambda}_tM_t\mathbb{E}[z^2]}\left[-\nu_2\Phi_t  - \left(\nu_2\bar{p}(t; t) - k_1\hat{\lambda}_t\int_{\mathbb{R}_+}zf(z)\bar{Z}_2(t, z;t)dz\right)\right]^+\\
			& = a^*_t - \frac{A^c_t}{k_1\hat{\lambda}_tM_t\mathbb{E}[z^2]}\left(\nu_2\bar{p}(t; t) - k_1\hat{\lambda}_t\int_{\mathbb{R}_+}zf(z)\bar{Z}_2(t, z;t)dz\right) \triangleq a^*_t + D^c_2(t), 
		\end{align*}
		where $0\leq A^c_t\leq 1$ is a bounded adapted process. To prove the uniqueness of the equilibrium control, we must demonstrate that $D^c_1(t) \equiv D^c_2(t)\equiv 0$ as follows.
		Here,
		\begin{align}\label{pbar0}
			\begin{split}
				d\bar{p}(s;t) &=\left\{ -r_s\bar{p}(s;t) - \nu_1D^c_1(s)M_s - \nu_2D^c_2(s)M_s + \mathbb{E}_t[ \nu_1D^c_1(s)M_s + \nu_2D^c_2(s)M_s]\right\}ds \\
				&+ \bar{Z}(s;t)^\top dW_s + \int_{\mathbb{R}_+}\bar{Z}_2(s, z;t)d\widetilde{N}(ds, dz),\\
				\bar{p}(T; t) &= 0, ~s \in [t, T].
			\end{split}
		\end{align}
		As the interest rate $r(\cdot)$ is a bounded deterministic function, we take $r \equiv 0$ without loss of generality. Taking the conditional expectation on both sides of \eqref{pbar0}, we get $\mathbb{E}_t[\bar{p}(s; t)] = 0$ for $s \geq t$. Specifically, $\bar{p}(t; t) =0$. Hence,  $D_1^c(t) = -\frac{\bar{Z}_1(t; t)}{M_t\sigma(t)}$, and $D_2^c(t)  = \frac{A^c_t}{M_t\mathbb{E}[z^2]}\int_{\mathbb{R}_+}zf(z)\bar{Z}_2(t, z; t)dz$.  Then, $\nu_1D_1^c(t)M_t = -\frac{\mu(t)-r_t}{\sigma(t)}\bar{Z}_1(t; t)$, and $\nu_2D_2^c(t)M_t = \frac{\eta\mu_zA^c_t}{\mathbb{E}[z^2]}\int_{\mathbb{R}_+}z\bar{Z}_2(t, z; t)\delta(dz)$.  According to Proposition \ref{prop2}, we obtain $Z(s; t_1) = Z(s; t_2)$ and $Z_2(s, z; t_1) = Z_2(s, z; t_2)$ for a.e. $s \geq \max(t_1, t_2)$. We define $\Delta_c(t)= -\theta_1(t)\bar{Z}_1(t; t)+ \int_{\mathbb{R}_+}\theta_2(t; z)\bar{Z}_2(t, z; t)\delta(dz)$ and $\tilde{p}(s; t) = \bar{p}(s; t) - \int_{s}^{T}\mathbb{E}_t[\Delta_c(u)]du$, where $\theta_1(s) = \frac{\mu(s)-r_s}{\sigma(s)}$ and $\theta_2(s, z) = \frac{\eta\mu_zA^c_t}{\mathbb{E}[z^2]}z \geq 0$.  Then, 
		\begin{align}
			\begin{split}
				d\tilde{p}(s;t) &=\left\{ - \nu_1D^c_1(s)M_s - \nu_2D^c_2(s)M_s \right\}ds + \bar{Z}(s;t)^\top dW_s + \int_{\mathbb{R}_+}\bar{Z}_2(s, z;t)d\widetilde{N}(ds, dz),\\
				\tilde{p}(T; t) &= 0, ~s \in [t, T].
			\end{split}
		\end{align}
		We then introduce a new measure, $\mathbb{Q}$, on $\mathcal{F}_t$ by $\frac{d\mathbb{Q}}{d\mathbb{P}}= \mathcal{E}_t$,  where 
		\begin{align}\label{mathcalE}
			\begin{split}
				\mathcal{E}_t &= \exp\bigg\{-\int_{0}^{t}\theta_1(s)dW_1(s) - \frac{1}{2}\int_{0}^{t}\theta_1(s)^2ds\\
				&+ \int_{0}^{t}\int_{\mathbb{R}_+}\ln(1+\theta_2(s, z))\widetilde{N}(ds, dz) + \int_{0}^{t}\int_{\mathbb{R}_+}\left\{\ln(1+\theta_2(s, z)) -\theta_2(s, z)\right\}\delta(dz)ds \bigg\}. 
			\end{split}
		\end{align}
		Note that if $\theta_2(s,z)\geq 0$, then $\ln(1 +\theta_2(s,z)) \leq \theta_2(s,z)$.
		If $C_2>k_1\frac{\eta^2\mu_z^2}{\mathbb{E}[z^2]}$ as in Assumption \ref{assume:lambda}, then the Novikov condition,
		\begin{align*}
			\mathbb{E}\left[\exp\left(\frac{1}{2}\int_{0}^{T}\theta_1(s)^2ds + \int_{0}^{T}\int_{\mathbb{R}_+}\{(1+\theta_2(s,z))\ln(1 + \theta_2(s,z)) - \theta_2(s,z)\}\delta(dz)ds \right) \right] < \infty,
		\end{align*}
		is satisfied. 
		Using measure $\mathbb{Q}$,  we obtain 
		\begin{equation}
			d\tilde{p}(s; t) =  \bar{Z}_0(s;t)dW_0(s) + \bar{Z}_1(s; t)dW^\mathbb{Q}_1(s) + \int_{\mathbb{R}_+}\bar{Z}_2(s, z;t)d\widetilde{N}^\mathbb{Q}(ds, dz),
		\end{equation}
		where $W^\mathbb{Q}_1(s)$ and $\widetilde{N}^\mathbb{Q}(ds, dz)$ are the standard Brownian motion and compensated Poisson random measure under $\mathbb{Q}$, respectively. 
		
		For a set $q_0 \in (1, 2)$, and given that $z$ has finite moments, according to Lemma \ref{lemma1}, there exists a constant $C$ such that
		\begin{align*}
			&\mathbb{E}\left[\sup_{s \in [t, T]}\left|\int_{s}^{T}\mathbb{E}_t[\Delta_c(u)] du\right|^{q_0} \right]
			\leq C\mathbb{E}\left[\left(\int_{t}^{T}|\Delta_c(u)|du\right)^{q_0}\right] \leq C \left\{\mathbb{E}\left[\int_{t}^{T}\bar{Z}_1(s; s)^2ds\right]\right\}^\frac{q_0}{2} \\
			&+ C\left\{\mathbb{E}\left[\int_{t}^{T}\int_{\mathbb{R}_+}\bar{Z}_2(s, z; s)^2\delta(dz)ds\right]\right\}^\frac{q_0}{2}\left\{\mathbb{E}\left[\int_{t}^{T}\int_{\mathbb{R}_+}\theta_2(s,z)^\frac{2q_0}{2-q_0}\delta(dz)ds\right]\right\}^\frac{2-q_0}{2} <\infty. 
		\end{align*}
		Hence, $\mathbb{E}\left[\sup_{t\leq s \leq T}|\tilde{p}(s; t)|^{q_0}\right]< \infty$.  
		For any $m \geq 1$, assume that $\mathbb{E}[\mathcal{E}_T^m]  < \infty$; then, for the aforementioned $q_0$, $q_1 \in (1, q_0)$ exists such that 
		\begin{align*}
			\mathbb{E}^\mathbb{Q}\left[\sup_{t\leq s \leq T}|\tilde{p}(s; t)|\right] \leq \left\{\mathbb{E}\left[\sup_{t\leq s \leq T}|\tilde{p}(s; t)|^{q_1}\right]\right\}^\frac{1}{q-1}\left\{\mathbb{E}\left[\mathcal{E}_T^\frac{q_1}{q_1-1}\right]\right\}^\frac{q_1-1}{q_1}< \infty. 
		\end{align*}
		If $m \geq 2$, then values of $q_0$ and $q_1$ exist such that the above boundedness holds. We analyze the condition  $ \mathbb{E}\left[ \mathcal{E}_T^2\right] < \infty$, where $\mathcal{E}$ is given in \eqref{mathcalE}. Note that $\theta_1$ is a bounded deterministic function and $\theta_2(s,z)\geq 0$. Therefore,
		\begin{align*}
			&\mathbb{E}[\mathcal{E}_T^2]  \leq\\
			&C\mathbb{E}\left[\exp\left(\int_{0}^{T}\int_{\mathbb{R}_+}2\ln(1+\theta_2(s, z))\widetilde{N}(ds, dz) + \int_{0}^{T}\int_{\mathbb{R}_+}2\left\{\ln(1+\theta_2(s, z)) -\theta_2(s, z)\right\}\delta(dz)ds\right)\right]\\
			&\leq C\mathbb{E}\left[\exp\left(\int_{0}^{T}\int_{\mathbb{R}_+}\ln(1+\theta_2(s, z))^2N(ds, dz)\right)\right]  = C\mathbb{E}\left[\exp\left(\sum_{i = 1}^{N_T}\ln(1+ \theta_2(t_i, z_i))^2  \right)\right]\\
			& = C\mathbb{E}\left[\mathbb{E}\left[\left.e^{\sum_{i =1}^{N_T}\ln(1+ \theta_2(t_i, z_i))^2}\right|N_T \right] \right] \leq C\mathbb{E}\left[\mathbb{E}\left[\left(1 +\frac{\eta\mu_z}{\mathbb{E}[z^2]}z\right)^2\right]^{N_T}\right]
		\end{align*}
		for a constant $C$. Let $\mathbb{E}\left[\left(1 +\frac{\eta\mu_z}{\mathbb{E}[z^2]}z\right)^2\right] \triangleq A_0$ for a constant $A_0$. Then, we obtain 
		\begin{align*}
			\mathbb{E}\left[\mathbb{E}\left[\left(1 +\frac{\eta\mu_z}{\mathbb{E}[z^2]}z\right)^2\right]^{N_T}\right] &= \mathbb{E}\left[e^{(\ln A_0)N_T}\right] = \mathbb{E}\left[\mathbb{E}\left[\left.e^{(\ln A_0)N_T}\right|\hat{\lambda}_{[0, T]}\right]\right] \\
			&= \mathbb{E}\left[\exp\left((A_0 -1)\int_{0}^{T}k_1\hat{\lambda}_sds\right) \right].
		\end{align*}
		As $\mu_z^2 \leq \mathbb{E}[z^2]$,  $A_0 \leq 1 + \frac{2\eta\mu_z^2}{\mathbb{E}[z^2]} + \frac{\eta^2\mu_z^2}{\mathbb{E}[z^2]}\leq(1 + \eta)^2$. There exists a constant value of $C$, such that $\mathbb{E}\left[\mathcal{E}_T^2\right] \leq C\mathbb{E}\left[\exp\left((2\eta+\eta^2)\int_{0}^{T}k_1\hat{\lambda}_sds\right) \right]$. If $C_2 \geq k_1(2 + \eta)\eta$ as in Assumption \ref{assume:lambda}, then $\mathbb{E}\left[\mathcal{E}_T^2\right] < \infty$.  
		Therefore, $\tilde{p}(s; t)$ is a $\mathbb{Q}$-martingale, 
		$\widetilde{p} \equiv \bar{Z} \equiv \bar{Z}_2 \equiv 0$, and  $D_1^c \equiv 0$ and $D_2^c \equiv 0$.  Following this, equilibrium control $(\pi^*, a^*)$, as derived by \eqref{u*c}, is unique. 
	\end{proof}
	
	From the above proof, we recognize that $C_2\geq k_1(2 + \eta)\eta$ is sufficient to confirm the admissibility and uniqueness of the equilibrium control \eqref{u*c} in a case with constant risk aversion.

	We next study the uniqueness of the equilibrium strategy under the interval reinsurance constraint. 
	\begin{theorem}
		Under Assumption \ref{assume:lambda}, the strategy in \eqref{u*constrain} represents the unique equilibrium control to Problem \eqref{objective}, where $\phi_1=0$ and $a\in [0,1]$ is the constraint.
	\end{theorem}
	\begin{proof}
		Suppose instead that there is an alternative admissible equilibrium control pair, $(\pi, a) \in  H_\mathcal{F}^2(0,T;\mathbb{R}, \mathbb{P})\times \cup_{q>2}L_\mathcal{F}^q(0,T;\mathcal{D}, \mathbb{P})$, with the corresponding state process $X$ and $\mathcal{D} = [0, 1]$. By replacing $X^*$ with $X$, the BSDE \eqref{p*} admits a unique solution $(p(s; t), Z(s; t), Z_2(s,z;t))$. This satisfies the condition \eqref{condition}, where $Z(s; t) =\left(Z_0(s; t), Z_1(s; t)\right)^\top$. We define
	 \begin{align*}
	 & \bar{p}(s;t) = p(s; t) - \left(M_sX_s + \Gamma_s^{(2)} - \mathbb{E}_t[M_sX_s + \Gamma_s^{(2)} - \Phi_s]\right),\\
	&\bar{Z}_0(s;t) = Z_0(s;t)- \gamma^{(2)}_0(s),~\bar{Z}_1(s; t) = Z_1(s;t) - M_s \pi(s)\sigma(s),\\
	&\bar{Z}_2(s,z;t) = Z_2(s,z;t) + M_sa(s)z,
	\end{align*}
	
	Using $M$, $\Gamma^{(2)}$, $\Phi$, and $\gamma^{(2)}_0$ as defined in  Section \ref{section:constrain}. Substituting the above definitions into \eqref{condition}, we obtain 
	\begin{align}
		\left\{
		\begin{array}{lr}
			\nu_1(t)\left[\bar{p}(t; t) +\Phi_t\right] + \sigma(t)\left[\bar{Z}_1(t; t) + M_t \pi(t)\sigma(t)  \right] = 0,\\ 
			\langle\nu_2(t)\left[\bar{p}(t; t) +\Phi_t\right]  - k_1\hat{\lambda}\int_{\mathbb{R}_+}z_tf(z) \left[\bar{Z}_2(t, z; t)- M_ta(t)z\right]dz, \rho_2(t) - a_t\rangle \geq 0
		\end{array}\right.
	\end{align}
	for any $\rho_2 \in \cup_{q>2}L_\mathcal{F}^q(t,T;\mathcal{D}, \mathbb{P})$. Therefore, 
	\begin{align*}
		\pi(t) &= -\frac{\nu_1}{M_t\sigma(t)^2} \Phi_t - \frac{1}{M_t\sigma(t)^2}(\nu_1(t)\bar{p}(t; t) + \sigma(t)\bar{Z}_1(t; t))\\
		& = \pi^*(t) - \frac{1}{M_t\sigma(t)^2}(\nu_1(t)\bar{p}(t; t) + \sigma(t)\bar{Z}_1(t; t)) \triangleq  \pi^*(t) + \tilde{D}^c_1(t), \\
		a(t) &= {\rm Proj}_\mathcal{D}\left\{ \frac{1}{k_1\hat{\lambda}_tM_t\mathbb{E}[z^2]}\left[-\nu_2\Phi_t  - \left(\nu_2\bar{p}(t; t) - k_1\hat{\lambda}_t\int_{\mathbb{R}_+}zf(z)\bar{Z}_2(t, z;t)dz\right)\right]^+\right\}\\
		& = {\rm Proj}_\mathcal{D}\left(-\frac{\eta \mu_z\Phi_t}{M_t\mathbb{E}[z^2]}\right) + {\rm Proj}_\mathcal{D}\left\{ - \frac{1}{k_1\hat{\lambda}_tM_t\mathbb{E}[z^2]}\left(\nu_2\bar{p}(t; t) - k_1\hat{\lambda}_t\int_{\mathbb{R}_+}zf(z)\bar{Z}_2(t, z;t)dz\right)\right\}\\
		& = a^*_t - \frac{\tilde{A}_t}{k_1\hat{\lambda}_tM_t\mathbb{E}[z^2]}\left(\nu_2\bar{p}(t; t) - k_1\hat{\lambda}_t\int_{\mathbb{R}_+}zf(z)\bar{Z}_2(t, z;t)dz\right) \triangleq a^*_t + \tilde{D}^c_2(t), 
	\end{align*}
	where $0 \leq \tilde{A}_t \leq 1$ is a bounded adapted process. By following a method similar to the proof of Theorem \ref{uniqueCRA}, we derive a unique control in \eqref{u*constrain} when $C_2 \geq k_1(2 + \eta)\eta$ in Assumption \ref{assume:lambda}. 
	
    \end{proof}
		
		Therefore, the equilibrium controls for the reinsurance policy situations with and without the boundedness constraint are unique in their respective settings. This finding is important because it covers both Markovian models and non-Markovian VMM situations with respect to the mortality rate. The resulting strategies are Markovian and independent of the historical mortality rate for insurers with constant risk aversion. The result for the state-dependent case is rather different, as shown in the following section.
		
		\section{Equilibrium strategy under state-dependent risk aversion}
		\label{Sec:SDRA}
		Consider a case of state-dependent risk aversion when the reinsurance company offers a cheap reinsurance premium, where $\phi_2 = 0$ and $\phi_1 >0$. In other words, $c \equiv 0 $ or, equivalently, $\eta = \theta$. Here, let $\mathcal{D} = [0, \infty)$. 
		
		Consider the following form for $p^*(s; t)$.
		\begin{equation}
			p^*(s;t) = M_sX^*_s - \Gamma_s^{(1)}X^*_t - \mathbb{E}_t[M_sX^*_s ], 
		\end{equation}
		where  $\Gamma_t^{(1)}  = \phi_1e^{\int_{t}^{T}r_sds}$,  $(M, U)$ is a solution to the BSDE, as follows:
		\begin{align}
			dM_s = -F_sds + U_s^\top dW_s, ~ M_T = 1, 
		\end{align}
		where $U(s) = \left(U_0(s), U_1(s)\right)^\top$ represents a 2-dimensional vector. 
		
		After suppressing the dependence of $s$, we repeat the procedure in a manner similar to that described in Section \ref{subsec:CRA}. We deduce that 
		\begin{align}
			\pi^*(t) = \frac{1}{M_t\sigma(t)} \left(\frac{\nu_1(t) \Gamma_t^{(1)}}{\sigma(t)} - U_1(t)\right)X^*_t,~
			a^*(t) = \frac{\nu_2(t)}{M_tk_1\hat{\lambda}_t\mathbb{E}[z^2]}\Gamma_t^{(1)}X^*_t,
		\end{align}
		where $(M, U)$ is a solution to the BSDE: 
		\begin{align*}
			\left\{
			\begin{array}{lr}
				dM_s = \left\{-2r_sM_s + \left(\frac{\nu_1}{\sigma} - \frac{\nu_1(s)\Gamma_s^{(1)}}{M_s\sigma(s)}\right)U_1(s)+ \frac{1}{M_s} U_1^2(s) - \left( \frac{\nu_1(s)^2}{\sigma(s)^2} + \frac{\nu_2(s)^2}{k_1\hat{\lambda}_s\mathbb{E}[z^2]}\right)\Gamma_s^{(1)}\right\}ds + U_s^\top dW_s,\\
				M_T =1.\\
			\end{array}
			\right.
		\end{align*}
		
		As $\nu_1 = \mu - r$ and $\nu_2 = \eta\mu_zk_1\hat{\lambda}$ are deterministic, $\hat{\lambda}$ is the only stochastic parameter among the coefficients of the BSDE \eqref{BSDE:M}. Therefore, we set $U_1(\cdot) = 0$ and show that the resulting BSDE admits a unique solution. Specifically, 
		\begin{align}\label{BSDE:M}
			\left\{
			\begin{array}{lr}
				dM_s = \left\{-2r_sM_s  - \left( \frac{\nu_1(s)^2}{\sigma(s)^2} + \frac{\nu_2(s)^2}{k_1\hat{\lambda}_s\mathbb{E}[z^2]}\right)\Gamma_s^{(1)}\right\}ds + U_0(s) dW_0(s),\\
				M_T =1.\\
			\end{array}
			\right.
		\end{align}
		
		\begin{pro}
			In this setting, the BSDE \eqref{BSDE:M} admits a unique solution $(M, U_0) \in S_\mathcal{F}^q(0,T;\mathbb{R}, \mathbb{P})\times H_\mathcal{F}^q(0,T;\mathbb{R}, \mathbb{P})$ for any $q>2$. 
		\end{pro}
		
		\begin{proof}
			In this setting, $r$ and $\frac{\nu_1}{\sigma}$ are deterministic functions. Here, $\frac{\nu_2^2}{k_1\hat{\lambda}\mathbb{E}[z^2]} = \frac{\eta^2\mu_z^2}{\mathbb{E}[z^2]}k_1\hat{\lambda}$. According to Lemma \ref{lemma1}, $\hat{\lambda} \in L^q_\mathcal{F}(0, T; \mathbb{R}, \mathbb{P})$ for any $q >2$. The result follows according to Theorem 5.1 in \cite{EPQ}. 
		\end{proof}
		
		\begin{pro}\label{pro:M}
			The explicit solution to the BSDE \eqref{BSDE:M} is given by
			\begin{eqnarray}
				M_t &=&  e^{\int_{t}^{T}2r_sds} + \int_{t}^{T}e^{\int_{t}^{s}2r_udu}\left( \frac{(\mu(s)-r_s)^2}{\sigma(s)^2} + \frac{\mu_z^2\eta^2}{ \mathbb{E}[z^2]}k_1\mathbb{E}[\hat{\lambda}_s|\mathcal{F}_t]\right)\Gamma_s^{(1)}ds,\label{M}  \cr
				U_0(t) &=& \int_{t}^{T}e^{\int_{t}^{s}2r_udu}\frac{\mu_z^2\eta^2}{\mathbb{E}[z^2]}k_1E_B(s-t)\sigma_\lambda\sqrt{\lambda_s}ds, \label{U0}
			\end{eqnarray}
			where 
			\begin{equation}\label{expect:lam}
				\mathbb{E}[\hat{\lambda}_s|\mathcal{F}_t] = l(s) + \left(1 - \int_{0}^{s}R_B(u)du\right)\lambda_0 + \int_{0}^{s}E_B(s-u)(b_1 - a_1\lambda_u)du + \int_{0}^{t}E_B(s- u)\sigma_{\lambda}\sqrt{\lambda_u}dW_0(u); 
			\end{equation}
			here, $B=-a_1$, $R_B$ is the resolvent of $-KB$, and $E_B = K- R_B*K$.
			Furthermore, $M_t \geq 1$ for $t \in [0, T]$. 
		\end{pro}
		\begin{proof}
			From \eqref{BSDE:M}, we can easily see that 
			\begin{align*}
				M_t &= \mathbb{E}\left[\left.  e^{\int_{t}^{T}2r_sds}  +  \int_{t}^{T}e^{\int_{t}^{s}2r_udu}\left( \frac{\nu_1(s)^2}{\sigma(s)^2} + \frac{\nu_2(s)^2}{k_1\hat{\lambda}_s\mathbb{E}[z^2]}\right)\Gamma_s^{(1)}ds\right|\mathcal{F}_t\right]\\
				& =  e^{\int_{t}^{T}2r_sds}  +  \int_{t}^{T}e^{\int_{t}^{s}2r_udu}\left( \frac{(\mu(s)-r_s)^2}{\sigma(s)^2} + \frac{\mu_z^2\eta^2}{ \mathbb{E}[z^2]}k_1\mathbb{E}[\hat{\lambda}_s|\mathcal{F}_t]\right)\Gamma_s^{(1)}ds.
			\end{align*}
			The result follows according to Lemma \ref{lemma:expmu}.
			
			As $\hat{\lambda} \geq 0$, the conditional expectation $\mathbb{E}[\hat{\lambda}_s|\mathcal{F}_t] \geq 0$ when $0 \leq t \leq s \leq T$. The interest rate $r > 0$. According to the representation of $M_t$ in \eqref{M}, we know that $M_t \geq 1$ for $t \in [0, T]$. 
		\end{proof}

		Although we obtain an explicit expression of the equilibrium control, we must prove its admissibility. Under state-dependent risk aversion, the equilibrium control depends on $X^*$. We also encounter a challenge posed by the unbounded Volterra process $\hat{\lambda}$. We must make some additional assumptions to prove admissibility. 
		
		\begin{assumption}\label{assume:U0}
			A sufficiently large constant $C_1$ exists, such that 
			\[\mathbb{E}\left[\exp\left(C_1\int_{0}^{T}U_0^2(s)ds\right) \right] < \infty. \]
		\end{assumption}
		\begin{assumption}\label{assume:lambda}
			A sufficiently large constant $C_2$ exists, such that 
			\[ \mathbb{E}\left[\exp\left(C_2\int_{0}^{T}\lambda_t dt\right)\right] < \infty. \]
		\end{assumption}
		
		Assumptions \ref{assume:U0} and \ref{assume:lambda} are similar to those made by \cite{HW19} and \cite{YW19, YW}. These assumptions are regularities for the unbounded parameters under Brownian filtration. However, we also encounter jumps. To enable the use of It\^o's calculus under a Poisson random measure, we must regulate the randomness of the claim size $z$.
		
		\begin{assumption} \label{assume:z}
			\[\phi_1\eta\mu_z\max\{z\}\leq \mathbb{E}[z^2].  \]
		\end{assumption}
		
		According to Assumption \ref{assume:z}, the claim size $z$ has an upper bound associated with the safety loading parameter $\eta$, which is often a small value.
		
		\begin{theorem}
			Based on Assumptions \ref{assume:U0}--\ref{assume:z}, there exists an admissible equilibrium control to Problem \eqref{objective} under state-dependent risk aversion, as follows:
			\begin{equation}\label{u^*state}
				\pi^*(t) = \frac{\nu_1(t)}{M_t\sigma(t)^2} \Gamma_t^{(1)}X^*_t, ~
				a^*(t) = \frac{\eta\mu_z}{M_t\mathbb{E}[z^2]}\Gamma_t^{(1)}X^*_t > 0,
			\end{equation}
			where $M$ is given in \eqref{M}. Furthermore,  $(\pi^*(t), a^*(t)) \in H_\mathcal{F}^2(0,T;\mathbb{R}, \mathbb{P})\times \cup_{q>2}L_\mathcal{F}^q(0,T;\mathcal{D}, \mathbb{P})$, $X \in S_\mathcal{F}^2(0,T;\mathbb{R}_+, \mathbb{P})$, and $a^*(t) \leq \frac{\eta\mu_z}{\mathbb{E}[z^2]}\phi_1X^*_t$. 
		\end{theorem} 
		\begin{proof}
			We need only prove the admissibility of the equilibrium control in \eqref{u^*state}. 
			Define $\pi^*(t) = \zeta_1(t)X^*_t$ and $a^*(t) = \zeta_2(t)X_t^*$, where $\zeta_1(t) = \frac{\nu_1(t)}{M_t\sigma(t)^2}\Gamma_t^{(1)}$, and $\zeta_2(t) = \frac{\eta\mu_z}{M_t\mathbb{E}[z^2]}\Gamma_t^{(1)}$. Note that $0<\frac{\Gamma_t^{(1)}}{M_t} \leq \phi_1$. According to Assumption \ref{assume:z}, we obtain $0\leq\zeta_2(t)z\leq 1$ for $t \in [0, T]$. 
			By substituting  $\pi^*(t)$ and $a^*(t)$ into \eqref{X} and applying It\^o's formula, we obtain 
			\begin{align*}
				d(\ln M_tX^*_t) = \left( - r_t - \frac{1}{2}\sigma(t)^2\zeta_1(t)^2 -\frac{U_0^2(t)}{2M_t^2}  \right)dt + \int_{\mathbb{R}_+}\left\{\ln\left(1 -\zeta_2(t)z \right) + \zeta_2(t)z \right\}\delta(dz)dt\\
				+\frac{U_0(t)}{M_t}dW_0(t) + \sigma(t)\zeta_1(t)dW_1(t) +  \int_{\mathbb{R}_+}\ln\left(1 -\zeta_2(t)z \right) \widetilde{N}(dt, dz). 
			\end{align*}
			Under standard Brownian motion $W$, $\mathcal{E}_t(h\cdot W) = \exp\left(\int_{0}^{t}h_sdW_s - \frac{1}{2}\int_{0}^{t}h_s^2ds \right)$. 
			Therefore, 
			\begin{align*}
				X^*_t = e^{\int_{0}^{t}r_udu}\frac{X_0M_0}{M_t}\mathcal{E}_t\left(\frac{U_0}{M}\cdot W_0\right)\mathcal{E}_t(\sigma\zeta_1\cdot W_1)\exp\bigg\{ \int_{0}^{t}\int_{\mathbb{R}_+}\left\{\ln\left(1 -\zeta_2(s)z \right) + \zeta_2(s)z \right\}\delta(dz)ds\\
				+\int_{0}^{t}\int_{\mathbb{R}_+}\ln\left(1 -\zeta_2(s)z \right) \widetilde{N}(ds, dz)\bigg\}. 
			\end{align*}
			Note that 
			\begin{align*}
				&\int_{0}^{t}\int_{\mathbb{R}_+}\ln\left(1 -\zeta_2(s)z \right)\delta(dz)ds
				+ \int_{0}^{t}\int_{\mathbb{R}_+}\ln\left(1 -\zeta_2(s)z \right) \widetilde{N}(ds, dz)\\
				&= \int_{0}^{t}\int_{\mathbb{R}_+}\ln\left(1 -\zeta_2(s)z \right) \widetilde{N}(ds, dz) = \sum_{i = 0}^{N(t)}\ln\left(1 -\zeta_2(t_i)z_i \right) \leq 0. 
			\end{align*}
			Hence, $$0 < \exp\left\{\int_{0}^{t}\int_{\mathbb{R}_+}\ln\left(1 -\zeta_2(t)z \right)\delta(dz)dt
			+ \int_{0}^{t}\int_{\mathbb{R}_+}\ln\left(1 -\zeta_2(t)z \right) \widetilde{N}(dt, dz)\right\} \leq 1.$$
			As $M_0$ is bounded, for any $m>1$, we obtain
			\begin{align*}
				&\mathbb{E}\left[\sup_{0\leq t\leq T}|X_t^*|^{2m}\right] \\
				&\leq C \mathbb{E}\left[\sup_{0\leq t \leq T}\mathcal{E}^{2m}_t\left(\frac{U_0}{M}\cdot W_0\right)\mathcal{E}_t^{2m}(\sigma\zeta_1\cdot W_1)\exp\left(\int_{0}^{t}\int_{\mathbb{R}_+}2m\zeta_2(s)z\delta(dz)dt\right) \right]\\
				& \leq C \Bigg\{\mathbb{E}\left[\sup_{0\leq t \leq T}\mathcal{E}^{4m}_t\left(\frac{U_0}{M}\cdot W_0\right)\right]\mathbb{E}\left[\sup_{0\leq t \leq T}\mathcal{E}_t^{8m}(\sigma\zeta_1\cdot W_1)\right]\\
				&\mathbb{E}\left[\sup_{0\leq t \leq T}\exp\left(\int_{0}^{t}\int_{\mathbb{R}_+}8m\zeta_2(s)z\delta(dz)dt\right)\right]\Bigg\}^\frac{1}{2},
			\end{align*}
			at a constant $C > 0$. 
			According to Doob's martingale maximum inequality, if $C_1 \geq 4m(8m-1)$ in Assumption \ref{assume:U0}, then
			\begin{align*}
				&\mathbb{E}\left[\sup_{0\leq t \leq T}\mathcal{E}^{4m}_t\left(\frac{U_0}{M}\cdot W_0\right)\right]\leq \left(\frac{4m}{4m -1} \right)^{4m}\mathbb{E}\left[\sup_{0\leq t \leq T}\mathcal{E}^{4m}_T\left(\frac{U_0}{M}\cdot W_0\right)\right]\\
				&\leq  \left(\frac{4m}{4m -1} \right)^{4m}\left\{\mathbb{E}\left[\mathcal{E}_T\left(\frac{8mU_0}{M}\cdot W_0\right)\right]\right\}^\frac{1}{2}\left\{\mathbb{E}\left[\exp\left(4m(8m -1)\int_{0}^{T}\frac{U_0^2(s)}{M_s^2}ds\right)\right]\right\}^\frac{1}{2}< \infty. 
			\end{align*}
			As $\sigma\zeta_1$ is a bounded deterministic function, we have $\mathbb{E}\left[\sup_{0\leq t \leq T}\mathcal{E}_t^{8m}(\sigma\zeta_1\cdot W_1)\right] < \infty$. Setting $C_2 \geq 8m\frac{\eta\mu_z^2\phi_1k_1}{\mathbb{E}[z^2]}$ in Assumption \ref{assume:lambda}, 
			\begin{align*}
				\mathbb{E}\left[\sup_{0\leq t \leq T}\exp\left(\int_{0}^{t}\int_{\mathbb{R}_+}8m\zeta_2(s)z\delta(dz)dt\right)\right] = \mathbb{E}\left[\exp\left(\int_{0}^{T}8m\frac{\eta\mu_z^2\Gamma_t^{(1)}}{M_t\mathbb{E}[z^2]}k_1\hat{\lambda}_tdt\right)\right]< \infty. 
			\end{align*}
			Hence, we obtain $X^*\in S_\mathcal{F}^m(0, T; \mathbb{R}_+, \mathbb{P})$ for any $m >1$, and $(\pi^*(t), a^*(t)) \in H_\mathcal{F}^2(0,T;\mathbb{R}, \mathbb{P})\times \cup_{q>2}L_\mathcal{F}^q(0,T;\mathcal{D}, \mathbb{P})$. 
		\end{proof}
		
		We further analyze Assumption \ref{assume:U0}. As $U_0$ is given in \eqref{U0}, for any constant $C_1$, 
		\begin{align*}
			\mathbb{E}\left[\exp\left(C_1\int_{0}^{T}U_0^2(s)ds\right) \right] = \mathbb{E}\left[\exp\left(C_1\int_{0}^{T}\left(\int_{t}^{T}e^{\int_{t}^{s}2r_udu}\frac{\mu_z^2\eta^2}{\mathbb{E}[z^2]}k_1E_B(s-t)\sigma_\lambda ds\right)^2\lambda_tdt\right)\right].
		\end{align*}
		Assumption \ref{assume:U0} holds when $C_2 \geq C_1\sup_{0\leq t \leq T}\left|\int_{t}^{T}e^{\int_{t}^{s}2r_udu}\frac{\mu_z^2\eta^2}{\mathbb{E}[z^2]}k_1E_B(s-t)\sigma_\lambda ds\right|^2$ in Assumption \ref{assume:lambda}. According to the above proof, $\mu_z^2 \leq \mathbb{E}[z^2]$; therefore, we have 
		$$C_2 > \max\left\{125\eta^4\sigma_\lambda^2\sup_{0\leq t \leq T}\left|\int_{t}^{T}e^{\int_{t}^{s}2r_udu}k_1E_B(s-t) ds\right|^2,  18\eta\phi_1k_1\right\}$$.
		This is sufficient for Assumption \ref{assume:lambda} to hold if $E_B$ is defined as in Proposition \ref{pro:M}. 
		

		\subsection{Uniqueness of the equilibrium control}
		We further establish the uniqueness of the equilibrium control under the condition of state-dependent risk aversion. A detailed proof is provided in Appendix \ref{Appendix:theoremunique}. 
		\begin{theorem}\label{uniqueSDRA}
			Suppose that Assumption \ref{assume:lambda} holds true. Let $M$ be defined by \eqref{M}. For the case involving state-dependent risk aversion, the control $u^* = (\pi^*, a^*)$, given by \eqref{u^*state}, is the unique equilibrium control for the RI Problem \eqref{objective} in which $\phi_1 = 0$ and $\phi_2>0$. 
		\end{theorem}
		
		By combining this theorem with the result of the analysis of admissibility, we discover that   
		\begin{align*}
			C_2\geq &\max\bigg\{k_1(2 + \eta)\eta,~ 125\eta^4\sigma_{\lambda}^2\sup_{0\leq t \leq T}\left|\int_{t}^{T}e^{\int_{t}^{s}2r_udu}k_1E_B(s-t)ds\right|^2, ~ 18\eta\phi_1k_1\bigg\}
		\end{align*}
		in Assumption \ref{assume:lambda} is sufficient to ensure both the admissibility and uniqueness of the equilibrium control in \eqref{u^*state} under state-dependent risk aversion, where $E_B$ is defined as in Proposition \ref{pro:M}. Although the value $125$ seems to be large, $\sigma_{\lambda}$ and the safety loading factor are usually very small. In other words, a large value of $C_2$ is not required in Assumption \ref{assume:lambda}. Below, we give a sufficient condition under which Assumption \ref{assume:lambda} holds.
		
		\begin{pro}\label{largeenough}
			If $a_1^2 - 2C\sigma_\lambda^2 > 0$, then 
			\[\mathbb{E}\left[\exp\left(C\int_{0}^{T}\lambda_sds\right)\right] < \infty. \]
		\end{pro}
		\begin{proof}
			If $a_1^2 - 2C\sigma_\lambda^2 > 0$, then according to Lemma A.1 in \cite{HW}, the Riccati equation,
			\begin{equation}\label{psi_C}
				\psi = (C - a_1\psi + \frac{1}{2}\sigma_\lambda^2\psi^2 ),
			\end{equation}
			has a unique global continuous solution over $[0, T]$. According to Theorem 2.4 in \cite{HW},
			\[\mathbb{E}\left[\exp\left(C\int_{0}^{T}\lambda_sds\right)\right] < \infty. \]
		\end{proof}
		
		Proposition \ref{largeenough} further clarifies that the value of constant $C$ is actually not very large.

		\section{Numerical study}
		\label{section:numerical}
		To demonstrate the influence of LRD on a reinsurance strategy, we compare reinsurance strategies under the LRD and Markovian mortality models. Using the VMM defined in \eqref{mortality}, we can easily see that the model reduces to a Markovian mortality model,
		\begin{equation}
			d\lambda_t = (b_1 - a_1\lambda_t)dt + \sigma_{\lambda}\sqrt{\lambda_t}dW_0, \nonumber,
		\end{equation}
		when $K\equiv 1$. Hence, the VMM actually contains a Markovian mortality model as a special case. By setting $K \equiv 1$, our result can also be applied to a Markovian case. We thus compare the equilibrium strategies under VMM with those under LRD and its Markovian counterparts by setting different values of $K$ and retaining the same values of other parameters. 
		
		In this section, we use the fractional $K(t) = \frac{t^{\alpha - 1}}{\Gamma(\alpha)}$ for the VMM such that the Hurst parameter $H = \alpha - \frac{1}{2}$. The VMM reflects the LRD feature for $\alpha >1$ and the Markovian feature for $\alpha = 1$. We focus on a population whose members are all aged 50 years at the time $t = 0$. To reflect the effect of LRD, the insurer is assumed to have access to historical mortality rate data of this population beginning at age 30. Let $T = 3$ years. Next, we simulate a sample path of the mortality rate from ages 30 to 53 years, as shown in  Figure \ref{subfig:lam}. For simplicity, we set $l(t) \equiv 0$. The values of the other parameters are given as follows: $\alpha = 1.33$, $b_1 = 0.15$, $a_1 = 0.5$, $\lambda_0 = 0.18$, and $\sigma_{\lambda} = 0.1$. 
		
		Based on this mortality rate path, we compare the reinsurance and investment strategies under the VMM and Markovian mortality model. In the constant risk aversion case, equation \eqref{u*c} reveals that the equilibrium strategies are the same under both models. Hence, the LRD feature of the mortality rate does not affect the RI equilibrium strategy under constant risk aversion. Under state-dependent risk aversion, the difference in equilibrium strategies deduced using \eqref{u^*state} depends on the process $M$ in \eqref{M}. In the expression of $M_t$, only $\mathbb{E}[\hat{\lambda}_s|\mathcal{F}_t]~ (s \geq t)$ differs between the two mortality models. The value of $\mathbb{E}[\hat{\lambda}_s|\mathcal{F}_t]$ is calculated using \eqref{expect:lam} under the VMM, and using $\mathbb{E}[\hat{\lambda}_s|\mathcal{F}_t] = l(s) + \lambda_te^{-a_1(s-t)} + \frac{b_1}{a_1}(1 - e^{-a_1(s-t)})$ under the Markovian mortality model. According to \eqref{expect:lam}, we recognize that the historical mortality rate enables an adjustment to the value of $\mathbb{E}[\hat{\lambda}_s|\mathcal{F}_t]$. Under the Markovian mortality model, no adjustment is made, and the value of $\mathbb{E}[\hat{\lambda}_s|\mathcal{F}_t]$ only depends on the current mortality rate. Based on this observation, we numerically compare the equilibrium strategies under the two models. The following parameter values are assigned: $k_1 = 10$, $\mu_z = 1$, $\mathbb{E}[z^2] = 1.2$,  $r \equiv 0.05$, $\eta \equiv 0.2$, $\mu = 0.07$, and $\sigma = 0.2$. The path of the risky asset $S$ on $[0, T]$ is simulated as shown in Figure \ref{subfig:asset}.   
		
		\begin{figure}[H]
			\centering
			\subfigure[]{\includegraphics[width= 7cm]{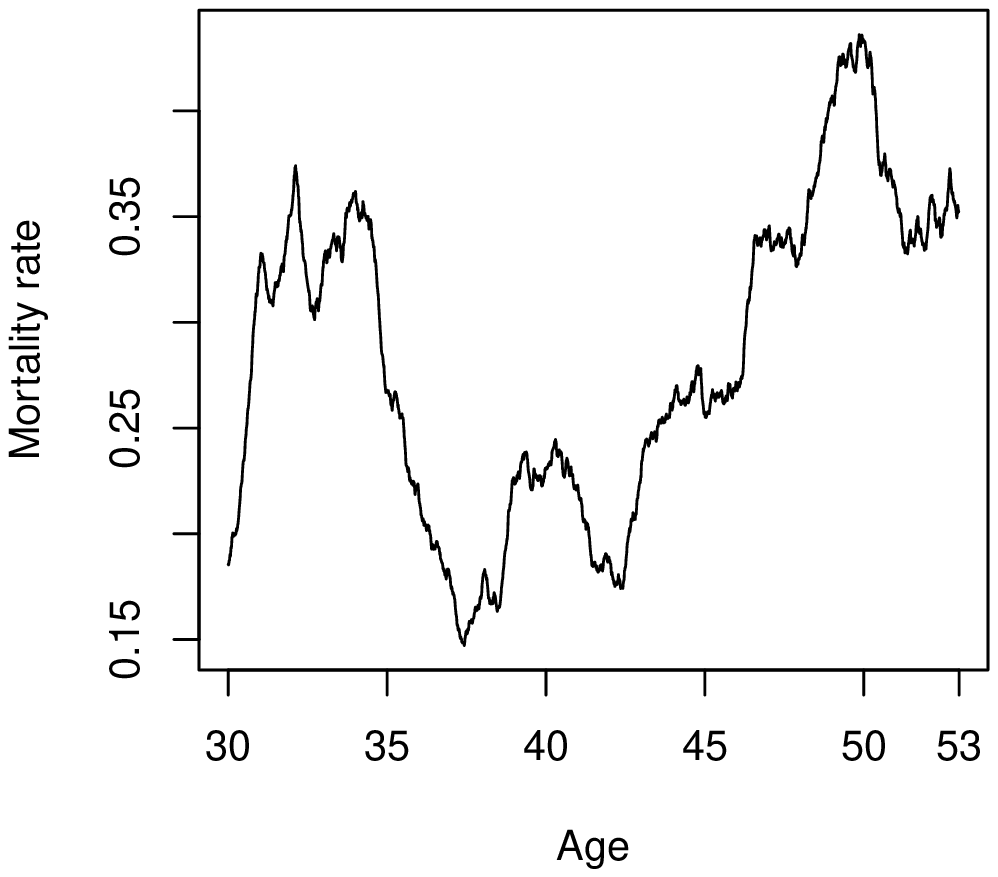} \label{subfig:lam}}
			\subfigure[]{\includegraphics[width = 7cm]{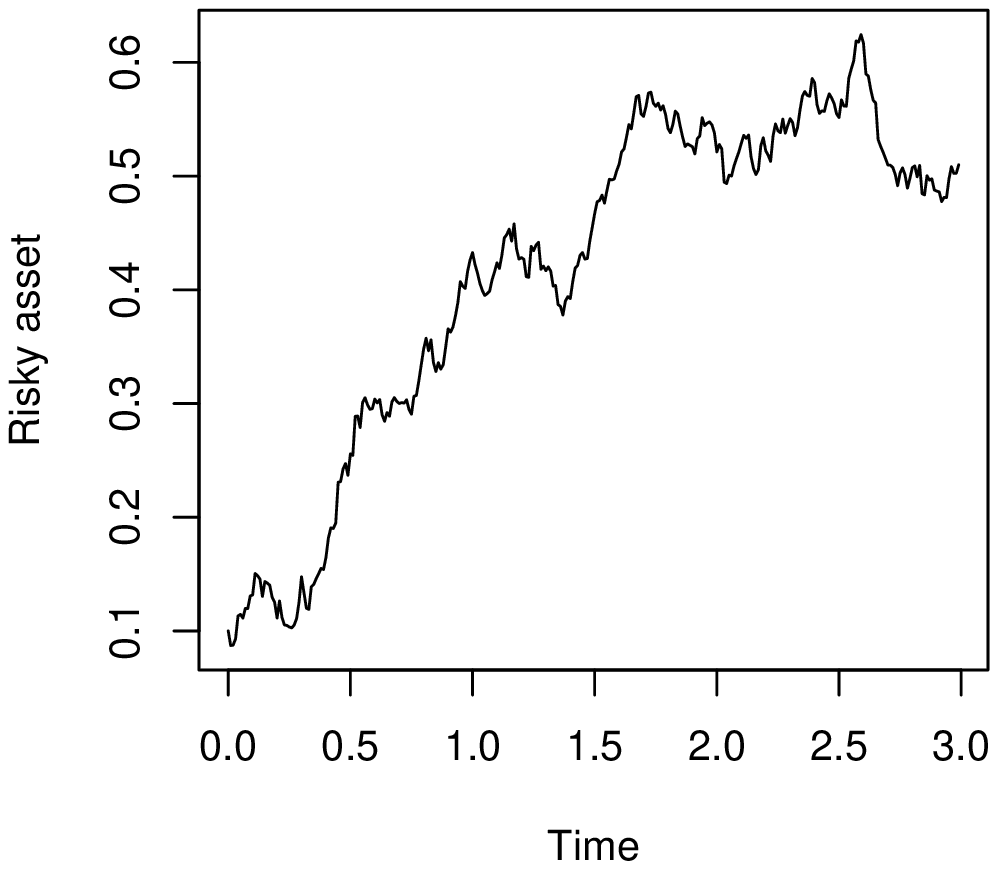}\label{subfig:asset}}
			\caption{Sample paths of the mortality rate and risky asset}
			\label{fig1}
		\end{figure}

		\begin{figure}[H]
			\centering
			\subfigure[]{\includegraphics[width= 7cm]{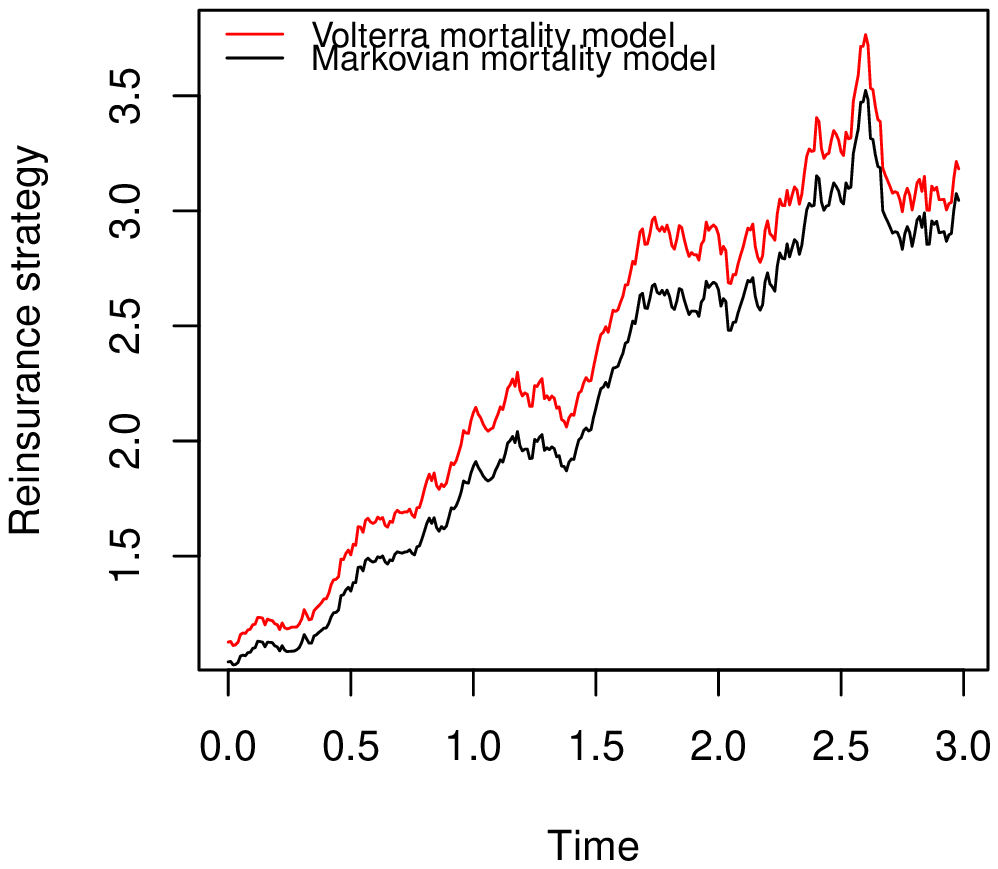}}
			\subfigure[]{\includegraphics[width = 7cm]{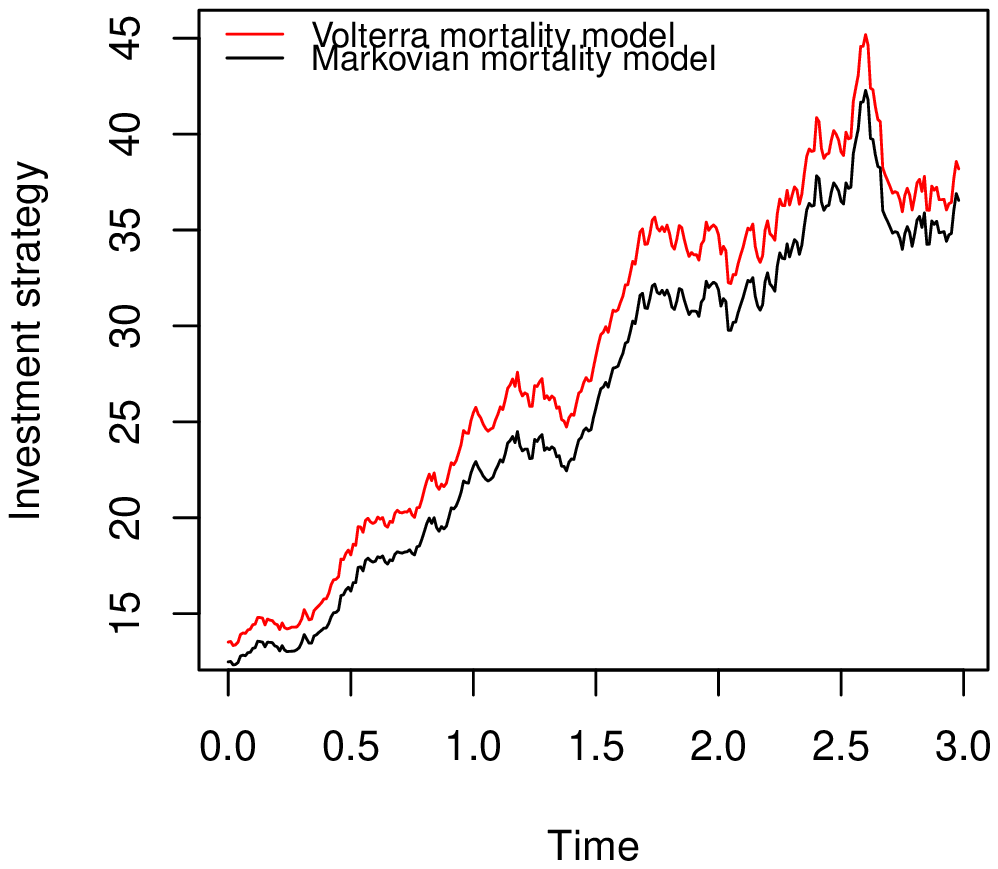}}
			\caption{Investment and reinsurance strategies under two mortality models}
			\label{fig2}
		\end{figure}
		
		We set an initial wealth value of $X_0 = 10$ at $t =0$ and let $\phi_1 = 1$. Then, we calculate the equilibrium strategies under the VMM and Markovian mortality model according to \eqref{u^*state}. Figures \ref{fig2} and \ref{fig3} plot the strategies and wealth processes under the two models. Figure \ref{fig2} shows that the LRD feature influences both the investment and reinsurance strategies. Figure \ref{fig3} shows that the LRD mortality model outperforms its Markovian counterpart once the mortality rate includes the LRD feature. 
		
		\begin{figure}[H]
			\centering
			\includegraphics[width= 7 cm ]{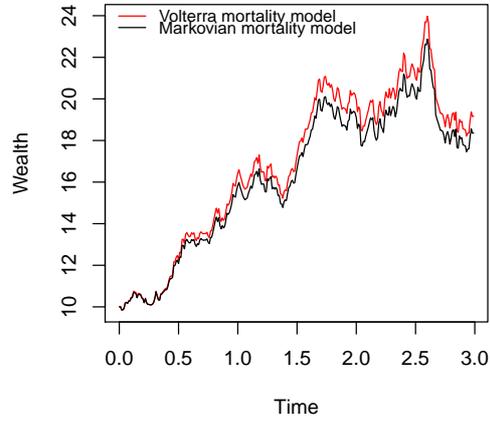}
			\caption{Wealth processes under the equilibrium strategies of the two models}
			\label{fig3}
		\end{figure}
		
		\begin{figure}[H]
			\centering
			\includegraphics[width= 7cm ]{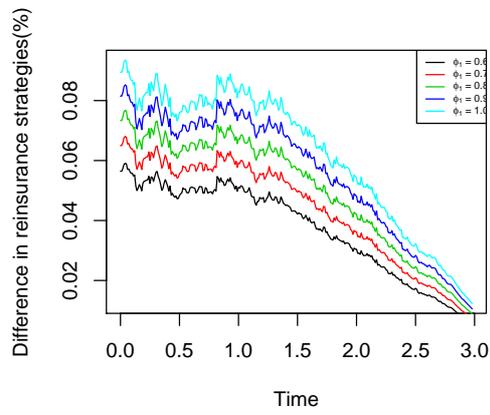}
			\caption{Percentage differences in the reinsurance strategies under the two models with different values of $\phi_1$}
			\label{fig4}
		\end{figure}
		
		\begin{figure}[H]
			\centering
			\includegraphics[width= 7cm ]{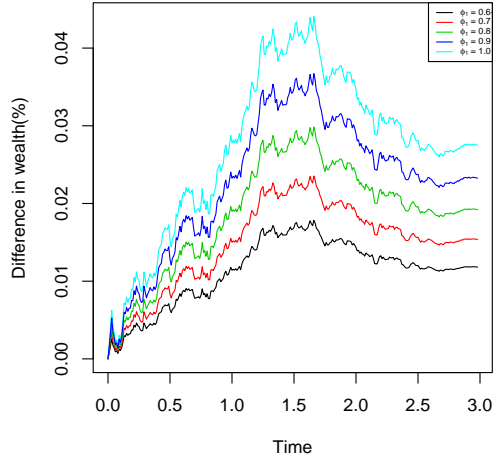}
			\caption{Percentage differences in wealth under the two models with different values of $\phi_1$}
			\label{fig5}
		\end{figure}
		
		To further analyze the influence of LRD, we vary the state-dependent risk aversion parameter $\phi_1$ from 0.5 to 1. We then plot the percentage differences in the reinsurance strategies and wealth processes under the two models in Figures \ref{fig4} and \ref{fig5}, respectively. These percentage differences both increase with $\phi_1$; at $\phi_1 = 1$, the differences in reinsurance strategies under both models could reach 10 percent, and the differences in the wealth level could reach 4 percent. Numerically, the LRD mortality model outperforms its Markovian counterparts by providing additional profit when LRD is included in the mortality rate.

		\section{Conclusion}
		\label{section:conclusion}
		Studies by \cite{WCW} and \cite{WW21} show that inclusion of the LRD feature in a mortality rate has a significant impact on longevity hedging. As a life risk management tool, reinsurance is a popular alternative to the longevity market. We show that the LRD feature in a mortality rate has a limited effect on equilibrium reinsurance strategies. In other words, risk management with reinsurance is more robust to the LRD. Specifically, under constant risk aversion, the equilibrium controls under the LRD and Markovian mortality models coincide with each other and the equilibrium control is unique. Under state-dependent risk aversion, LRD has a mild level of influence, which increases with the risk aversion level. Technically, this paper provides a mathematical solution to the time-consistent mean-variance reinsurance-investment problem for a mortality rate that includes LRD. We derive explicit open-loop equilibrium strategies for both constant and state-dependent risk aversion cases. By using unbounded parameters and imposing some technical conditions, we also prove the admissibility and uniqueness of the equilibrium controls in these two cases.
		
		\begin{appendices}
			\section{Some Proofs}
			\subsection{Proof of admissibility in Remark \ref{remark:admissible}}
			\label{Appendix:remarkproof}
			\begin{proof}
				From \eqref{X}, we have 
				\begin{align*}
					X_t  &= e^{\int_{0}^{t}r_u du}X_0 + \int_{0}^{t}e^{\int_{s}^{t}r_udu}\left(\nu(s)^\top u(s) + c_s\right)ds \\
					&+ \int_{0}^{t}e^{\int_{s}^{t}r_udu}\pi(s)\sigma(s)dW_1(s) - \int_{0}^{t}\int_{\mathbb{R}_+}e^{\int_{s}^{t}r_udu}a(s)z\widetilde{N}(ds,dz). 
				\end{align*}	
				Using the Burkholder--Davis--Gundy (BDG) inequality and H\"older's inequality and Lemma \ref{lemma1}, for any constant $q>2$,  there exists a constant $C> 0$ such that 
				\begin{align*}
					\mathbb{E}\left[\sup_{0\leq t \leq T}|X_t|^2\right] &\leq C \left\{1 +  \mathbb{E}\left[\int_{0}^{T}\pi_s^2ds\right] + \mathbb{E}\left[\left(\int_{0}^{T}\hat{\lambda}_s a_sds\right)^2\right] + \mathbb{E}\left[\int_{0}^{T}\hat{\lambda}_sa_s^2ds\right] \right\}\\
					& \leq C \left\{1 + \mathbb{E}\left[\int_{0}^{T}\pi_s^2ds\right] +   \left(\mathbb{E}\left[\int_{0}^{T}a_s^qds\right]\right)^{\frac{2}{q}} \left(\mathbb{E}\left[\int_{0}^{T}|\hat{\lambda}_s|^{\frac{2q}{q-2}}ds\right]\right)^{1-\frac{2}{q}} \right\}
					< \infty. 
				\end{align*}
			\end{proof}
			
			\subsection{Proof of Theorem \ref{theorem1} and Proposition \ref{prop1}}
			\label{appendix:Theo2prop1}
			\begin{proof}
				Denote by $X^\epsilon(\cdot)$ the state process corresponding to  $u^{t,\epsilon, \rho}(s) = ( \pi_s^{t, \epsilon, \rho_1}, a_s^{t, \epsilon, \rho_2} )^\top$. Using the standard perturbation approach, 
				\[X^\epsilon_s = X^*_s + Y^\epsilon(s) + Z^\epsilon(s),~ s\in[t, T] \]
				where $Y^\epsilon(s)$ and $Z^\epsilon(s)$ are the respective solutions to the following SDEs: 
				\begin{align*}
					\left\{
					\begin{array}{lr}
						dY^\epsilon(s) = r_sY^\epsilon(s)ds + \sigma\rho_1\boldsymbol{1}_{[t, t+\epsilon]}(s)dW_1(s) - \int_{\mathbb{R}_+}z(\rho_2- a^*) \boldsymbol{1}_{[t, t+\epsilon]}(s)\widetilde{N}(ds, dz), \\
						Y^\epsilon(t) =0, ~ s\in [t, T],  
					\end{array}\right.
				\end{align*}
				\[dZ^\epsilon(s) = r_sZ^\epsilon(s) +  \nu^\top\rho\boldsymbol{1}_{[t, t+\epsilon]}(s)ds, ~ Z^\epsilon(t) =0, ~ s\in [t, T]. \]
				According to the BDG inequality, there exists a positive constant $C$, such that 
				\begin{align*}
					&\mathbb{E}_t[\sup_{s \in [t, T]}(Y_s^\epsilon)^2]\\ 
					&= \mathbb{E}_t\left[\sup_{s\in[t, T]}\left(\int_{t}^{s}e^{\int_{v}^{s}r_udu}\sigma \rho_1\boldsymbol{1}_{[t, t+\epsilon]}dW_1(v) - \int_{t}^{s}\int_{\mathbb{R}_+}e^{\int_{v}^{s}r_udu}z(\rho_2 - a^*)\boldsymbol{1}_{[t, t + \epsilon]}(s)\widetilde{N}(dv, dz) \right)^2\right] \\
					&\leq C\mathbb{E}_t\left[\int_{t}^{T} \rho_1^2\boldsymbol{1}_{[t, t+\epsilon]}dv\right] + C\mathbb{E}_t\left[\int_{0}^{T}(\rho_2 - a^*_s)^2\hat{\lambda}_s\boldsymbol{1}_{[t, t+\epsilon]}ds \right]= O(\epsilon). 
				\end{align*}
				Moreover, 
				\begin{align}
					&\mathbb{E}_t[\sup_{s \in [t, T]}(Z_s^\epsilon)^2] = \mathbb{E}_t\left[\sup_{s\in[t, T]}\left(\int_{t}^{s}e^{\int_{v}^{s}r_udu}\nu^\top\rho\boldsymbol{1}_{[t, t+\epsilon]}dv\right)^2\right]\notag\\
					&\leq C \mathbb{E}_t\left[\left(\int_{t}^{t+\epsilon}|\nu^\top\rho| dv\right)^2\right] = O(\epsilon^2). 
				\end{align}
				The result is obtained by applying the same method as in \cite{SG}. 
			\end{proof}
			
			\subsection{Proof of Proposition \ref{prop2}} \label{proofP2}
			\begin{proof}
				Set $\widetilde{p}(s; t) = e^{-\int_{s}^{T}r_udu}p^*(s; t) + \mathbb{E}_t[X^*_T] + \phi_1X_t^* + \phi_2$,   $\widetilde{Z}(s; t) = e^{-\int_{s}^{T}r_udu}Z^*(s; t)$, and   $\widetilde{Z}_2(s, \cdot; t) = e^{-\int_{s}^{T}r_udu}Z^*_2(s, \cdot, t)$. According to It\^o's lemma, we obtain
				\begin{equation}\label{ptilde}
					d\widetilde{p}(s; t)  =  \widetilde{Z}(s;t)dW_s + \int_{\mathbb{R}_+} \widetilde{Z}_2(s,z;t) \widetilde{N}(ds,dz)  , ~ \widetilde{p}(T; t) = X^*_{T}. 
				\end{equation}
				Note that neither the coefficients nor the terminal condition of the above equation depend on $t$. \eqref{ptilde} can be regarded as a BSDE on the entire time interval $[0, T]$. For $s \in [0, T]$, denote the solution of \eqref{ptilde} as $(\widetilde{p}(s), \widetilde{Z}(s), \widetilde{Z}_2(s,\cdot)) \in S_\mathcal{F}^2(t,T;\mathbb{R}, \mathbb{P}) \times L_\mathcal{F}^2(t,T;\mathbb{R}^2, \mathbb{P}) \times F^2(t,T;\mathbb{R})$. Given the uniqueness of the solution, for any $t \in [0, T]$, $(\widetilde{p}(s; t), \widetilde{Z}(s; t), \widetilde{Z}_2(s,\cdot; t)) = (\widetilde{p}(s), \widetilde{Z}(s), \widetilde{Z}_2(s,\cdot))$. Then, the first claim in this lemma follows. 
				
				Using the definition of $\widetilde{p}(s; t)$, we obtain  
				\[{p}^*(s; t) = e^{\int_{s}^{T}r_udu}\widetilde{p}(s) - e^{\int_{s}^{T}r_udu}(\mathbb{E}_t[X^*_T] + \phi_1X_t^* + \phi_2) =e^{\int_{s}^{T}r_udu}\widetilde{p}(s) + e^{\int_{s}^{T}r_udu}\xi(t),\]
				where $\xi(t) = -(\mathbb{E}_t[X^*_T] + \phi_1X_t^* + \phi_2)$, and $\xi(\cdot) \in S_\mathcal{F}^2(t,T;\mathbb{R}, \mathbb{P})$. 
				Then, 
				\[\Lambda(s; t) = \Lambda_0(s) +e^{\int_{s}^{T}r_udu}\nu(s)\xi(t),\]
				where $\Lambda_0(s)  = \nu(s)e^{\int_{s}^{T}r_udu}\widetilde{p}(s)  +e^{\int_{s}^{T}r_udu}\left(\sigma(s)\widetilde{Z}_1(s),  k_1\hat{\lambda}_s\int_{\mathbb{R}_+}zf(z)\widetilde{Z}_2(s; z)dz\right)^\top$. 
			\end{proof}
			\subsection{Proof of Theorem \ref{theorem2}}
			\label{appendix:theorem3}
			\begin{proof}
				First, to prove sufficiency, we recall the representation $ \Lambda(s; t) = \Lambda_0(s) + e^{\int_{s}^{T}r_vdv}\nu(s)\xi(t)$ from Proposition \ref{prop2}. We still set $\rho_s = (\rho_1, \rho_2- a^*_s)^\top$.  Then,
				\[ \frac{1}{\epsilon}\int_{t}^{t+\epsilon}\mathbb{E}_t[\langle\Lambda(s; t), \rho_s\rangle]ds -  \frac{1}{\epsilon}\int_{t}^{t+\epsilon}\mathbb{E}_t[\langle \Lambda(s; s), \rho_s \rangle]ds  =  \frac{1}{\epsilon}\int_{t}^{t+\epsilon}\mathbb{E}_t\left[e^{\int_{s}^{T}r_vdv}\langle \nu(s), \rho_s\rangle(\xi(t)- \xi(s))\right]ds. \]
				Hence, we obtain 
				\begin{equation}\label{limLam}
					\liminf_{\epsilon\downarrow 0}\left|\frac{1}{\epsilon}\int_{t}^{t+\epsilon}\mathbb{E}_t[\langle\Lambda(s; t), \rho_s\rangle]ds -  \frac{1}{\epsilon}\int_{t}^{t+\epsilon}\mathbb{E}_t[\langle\Lambda(s; s), \rho_s\rangle]ds\right| = 0.
				\end{equation}
				If the condition in \eqref{condition} is satisfied, then 
				\[\liminf_{\epsilon\downarrow 0}\frac{1}{\epsilon}\int_{t}^{t+\epsilon}\mathbb{E}_t[\langle\Lambda(s; t), \rho_s\rangle]ds = \liminf_{\epsilon\downarrow 0}\frac{1}{\epsilon}\int_{t}^{t+\epsilon}\mathbb{E}_t[\langle\Lambda(s; s), \rho_s\rangle]ds \geq 0.\]
				According to Proposition \ref{prop1}, $(\pi^*, a^*)$ is an open-loop equilibrium control.  
				
				Second, we prove the necessity. If $u^* = (\pi^*, a^*)^\top$ is an open-loop equilibrium control, then according to Definition \ref{def2} and the variational equation in Theorem \ref{theorem1}, 
				\[\liminf_{\epsilon \downarrow0}\int_{t}^{t + \epsilon} \mathbb{E}_t \left[\langle \Lambda(s; t), \rho_s \rangle + \Theta(s)\langle \rho_s, \rho_s \rangle\right] ds \geq 0,  \]
				where $\rho$ is defined as in Theorem \ref{theorem1} and Proposition \ref{prop1}.
				Let $\rho_2 = 0$; then, the first condition in \eqref{condition} is a direct result of Theorem 3.1 in \cite{SG}. 
				For the second condition in \eqref{condition}, let $\rho_1=0$. According to Theorem \ref{theorem1}, 
				\begin{align*}
					\liminf_{\epsilon \downarrow 0}\int_{t}^{t+\epsilon}\mathbb{E}_t\bigg[\left(\nu_2(s)p^*(s; t)  - \int_{\mathbb{R}_+}zk_1\hat{\lambda}_tf(z) Z^*_2(s, z; t)dz\right)\widetilde{\rho}_2(s) \\
					+ \frac{1}{2}e^{\int_{s}^{T} 2r_u du }\left(\sigma(s)^2 + k_1\hat{\lambda}_s\mathbb{E}[z^2]\right)\widetilde{\rho}_2(s)^2\bigg]ds\geq 0, 
				\end{align*}
				where $\widetilde{\rho}_2  = \rho_2 - a^*\in \cup_{q>2}L_{\mathcal{F}}^q(0, T, \mathbb{R}, \mathbb{P})$.  Thus, 
				\begin{align*}
					\liminf_{\epsilon \downarrow 0}\int_{t}^{t+\epsilon}\mathbb{E}_t\bigg[\left(\nu_2(s)p^*(s; t)  - \int_{\mathbb{R}_+}zk_1\hat{\lambda}_tf(z) Z^*_2(s, z; t)dz\right){\rm sgn}(\widetilde{\rho}_2(s)) \\
					+ \frac{1}{2}e^{\int_{s}^{T} 2r_u du }\left(\sigma(s)^2 + k_1\hat{\lambda}_s\mathbb{E}[z^2]\right)|\widetilde{\rho}_2(s)|\bigg]ds\geq 0. 
				\end{align*}
				According to Lemma \ref{lemma1} and H\"older's inequality, 
				\[ \mathbb{E}_t\left[\int_{0}^{T}(\hat{\lambda}_sp^*(s;t))^2ds\right] \leq \left\{\mathbb{E}_t\left[\int_{0}^{T}p^*(s; t)^2\right]\right\}^{\frac{q}{2}} \left\{\mathbb{E}_t\left[\int_{0}^{T}\hat{\lambda}_s^{\frac{2q}{2 -q}}\right]\right\}^{1-\frac{q}{2}}<\infty,   \] 
				at a constant $q>2$. 
				Similarly,  $\mathbb{E}_t\left[\int_{0}^{T}(\hat{\lambda}_s\widetilde{\rho}_2(s))^2ds\right] < \infty$. 
				Thus, according to Lemma 3.5 in \cite{HHLb}, 
				$$\left(\nu_2(t)p^*(t; t)  - \int_{\mathbb{R}_+}zk_1\hat{\lambda}_tf(z) Z^*_2(t, z; t)dz\right)\widetilde{\rho}_2(t) + \frac{1}{2}e^{\int_{t}^{T} 2r_s ds }\left(\sigma(t)^2 + k_1\hat{\lambda}_t\mathbb{E}[z^2]\right)\widetilde{\rho}_2(t)^2\geq 0.$$
				For any $\theta \in (0, 1]$,  we use the same method as in \cite{HHLb} to obtain
				$$\left(\nu_2(t)p^*(t; t)  - \int_{\mathbb{R}_+}zk_1\hat{\lambda}_tf(z) Z^*_2(t, z; t)dz\right)\widetilde{\rho}_2(t) + \frac{1}{2}e^{\int_{t}^{T} 2r_s ds }\left(\sigma(t)^2 + k_1\hat{\lambda}_t\mathbb{E}[z^2]\right)\theta\widetilde{\rho}_2(t)^2\geq 0.$$
				Let $\theta\rightarrow 0^+$; we thus obtain
				$$\left(\nu_2(t)p^*(t; t)  - \int_{\mathbb{R}_+}zk_1\hat{\lambda}_tf(z) Z^*_2(t, z; t)dz\right)(\rho_2(t) - a^*(t))\geq 0.$$
			\end{proof}
			
			\subsection{Proof of Theorem \ref{uniqueSDRA}}
			\label{Appendix:theoremunique}
			\begin{proof}
				Suppose another equilibrium control $ u(\cdot) = (\pi(\cdot), a(\cdot))$ with the corresponding state process $X(\cdot)$. Then, when $X^*$ is replaced by $X$, the BSDE \eqref{p*} admits a unique solution $(p(s; t), Z(s; t), Z_2(s,z;t)) \in L_\mathcal{F}^2(t,T;\mathbb{R}, \mathbb{P})\times H_\mathcal{F}^2(t,T;\mathbb{R}^2, \mathbb{P})\times F^2(t,T;\mathbb{R})$. This satisfies the condition \eqref{condition}, where $Z(s; t) =\left(Z_0(s; t), Z_1(s; t)\right)^\top$.
				Then, we define
				\begin{align*}
					& \bar{p}(s;t) = p(s; t) - \left(  M_sX_s - \Gamma_s^{(1)}X_t - \mathbb{E}_t[M_sX_s]\right), \\
					&\bar{Z}_0(s;t) = Z_0(s;t) - X_sU_0(s),~\bar{Z}_1(s; t) = Z_1(s;t) - M_s \pi(s)\sigma(s) ,\\
					&\bar{Z}_2(s,z;t) = Z_2(s,z;t) + M_sa(s)z,
				\end{align*}
				where $M$, $U_0$, and $\Gamma^{(1)}$ are defined in Section \ref{Sec:SDRA}. 
				
				Clearly,  $(\bar{p}(s; t), \bar{Z}(s; t), \bar{Z}_2(s,z;t)) \in L_\mathcal{F}^2(t,T;\mathbb{R}, \mathbb{P})\times H_\mathcal{F}^2(t,T;\mathbb{R}^2, \mathbb{P})\times F^2(t,T;\mathbb{R})$. Similar to the proof of Theorem \ref{uniqueCRA}, we obtain
				\begin{align}
					\left\{
					\begin{array}{lr}
						\nu_1(t)\left[\bar{p}(t; t) -\Gamma_t^{(1)} X_t \right] + \sigma(t)\left[\bar{Z}_1(t; t) + M_t \pi(t)\sigma(t)  \right] = 0,\\ 
						\langle\nu_2(t)\left[\bar{p}(t; t) -\Gamma_t^{(1)}X_t\right]  - \int_{\mathbb{R}_+}zk_1\hat{\lambda}_tf(z) \left[\bar{Z}_2(t, z; t)- M_ta(t)z\right]dz, \rho_2(t) - a_t\rangle \geq 0
					\end{array}\right.
				\end{align}
				for any $\rho_2 \in \cup_{q>2}L_\mathcal{F}^q(t,T;\mathbb{R}_+, \mathbb{P})$. Then, 
				\begin{align*}
					\pi(t) &= \frac{\nu_1}{M_t\sigma(t)^2}\Gamma_t^{(1)}X_t - \frac{1}{M_t\sigma(t)^2}(\nu_1(t)\bar{p}(t; t) + \sigma(t)\bar{Z}_1(t; t))\\
					& = \pi^*(t) - \frac{1}{M_t\sigma(t)^2}(\nu_1(t)\bar{p}(t; t) + \sigma(t)\bar{Z}_1(t; t)) \triangleq  \pi^*(t) + D_1^d(t), \\
					a(t) &= \frac{1}{k_1\hat{\lambda}_tM_t\mathbb{E}[z^2]}\left[\nu_2\Gamma_t^{(1)}X_t - \left(\nu_2\bar{p}(t; t) - k_1\hat{\lambda}_t\int_{\mathbb{R}_+}zf(z)\bar{Z}_2(t, z;t)dz\right)\right]^+\\
					& = a^*_t - \frac{A^d_t}{k_1\hat{\lambda}_tM_t\mathbb{E}[z^2]}\left(\nu_2\bar{p}(t; t) - k_1\hat{\lambda}_t\int_{\mathbb{R}_+}zf(z)\bar{Z}_2(t, z;t)dz\right) \triangleq a^*_t + D_2^d(t), 
				\end{align*}
				where $0\leq A^d_t\leq 1$ is a bounded adapted process. Next, we show $D^d_1(t) \equiv D^d_2(t)\equiv 0$ to prove the uniqueness of the equilibrium control. We obtain
				\begin{align}\label{pbar}
					\begin{split}
						d\bar{p}(s;t) &=\left\{ -r_s\bar{p}(s;t) - \nu_1D^d_1(s)M_s - \nu_2D^d_2(s)M_s + \mathbb{E}_t[ \nu_1D^d_1(s)M_s + \nu_2D^d_2(s)M_s]\right\}ds \\
						&+ \bar{Z}(s;t)^\top dW_s + \int_{\mathbb{R}_+}\bar{Z}_2(s, z;t)d\widetilde{N}(ds, dz), \\
						\bar{p}(T; t) &= 0, s\in [t, T].
					\end{split}
				\end{align}
				As the interest rate $r(\cdot)$ is a bounded deterministic function, we take $r \equiv 0$ without a loss of generality. By taking the conditional expectation on both sides of \eqref{pbar},  we obtain $\mathbb{E}_t[\bar{p}(s; t)] = 0$ at $s \geq t$. Particularly, $\bar{p}(t; t) =0$. Hence,  $D^d_1(t) = -\frac{\bar{Z}_1(t; t)}{M_t\sigma(t)}$ and $D^d_2(t)  = \frac{A^d_t}{M_t\mathbb{E}[z^2]}\int_{\mathbb{R}_+}zf(z)\bar{Z}_2(t, z; t)dz$. 
				Then, $\nu_1D^d_1(t)M_t = -\frac{\mu(t)-r_t}{\sigma(t)}\bar{Z}_1(t; t)$ and $\nu_2D^d_2(t)M_t = \frac{\eta\mu_zA^d_t}{\mathbb{E}[z^2]}\int_{\mathbb{R}_+}z\bar{Z}_2(t, z; t)\delta(dz)$.  
				
				Following the same method used in the proof of Theorem \ref{uniqueCRA}, we obtain $\bar{Z} \equiv \bar{Z}_2 \equiv 0$. Thus, $D^d_1 \equiv D^d_2 \equiv 0$. As a result, $(\pi^*, a^*)$, given by \eqref{u^*state}, is the unique open-loop equilibrium control when $C_2 \geq k_1(2 + \eta)\eta$ in Assumption \ref{assume:lambda}. 
			\end{proof}
			
		\end{appendices}

	\end{document}